\documentclass[11pt]{article}
\usepackage{amsmath,amssymb,latexsym,amsthm,mathtools}
\usepackage[all]{xy}
\usepackage{float}
\usepackage{graphicx,epstopdf,url}
\usepackage{amsfonts}
\usepackage{appendix}
\usepackage{xcolor}
\usepackage[hidelinks]{hyperref}
\usepackage{booktabs}
\usepackage{multirow}
\usepackage{natbib}
\usepackage{bm}
\usepackage{algorithm}
\usepackage{algorithmic}
\newtheorem{theorem}{Theorem}

\usepackage[left]{lineno}
\usepackage[margin=1in]{geometry}
\usepackage{marginnote} \usepackage{marginfix}
\newcommand{\cb}[1]{{\color{black}#1}}
\newcommand{\mbh}[1]{{\color{black}#1}}

\title{Dynamical Survival Analysis for Modeling Hazard Functions with Nonlinear Systems} 

\author{Dananjani Liyanage
\thanks{Department of Mathematics and Statistics, University of Minnesota Duluth, dmadiwal@d.umn.edu} \and
    Mahmudul Bari Hridoy\thanks{Department of Mathematics, Virginia Tech, barihridoy@vt.edu} \and
    Fahad Mostafa\thanks{School of Mathematical and Natural Sciences, and Julie Ann Wrigley Global Futures Laboratory, Arizona State University, fahad.mostafa@asu.edu}
    \thanks{Corresponding author's email: \texttt{fahad.mostafa@asu.edu/fahadraj.du@gmail.com}}
}

\date{}
\begin{document}

\maketitle
\begin{abstract}
Hazard functions play a central role in survival analysis, providing insight into the underlying risk dynamics of time-to-event data, with broad applications in medicine, epidemiology, and related fields. First-order ordinary differential equation (ODE) formulations of the hazard function have been explored as extensions beyond classical parametric models. However, such approaches typically produce monotonic hazard patterns, limiting their ability to represent oscillatory behavior, nonlinear damping, or coupled growth–decay dynamics. We propose a \cb{general} statistical framework for modeling and simulating hazard functions governed by higher-order ODEs, allowing the hazard to depend on both its current level, its rate of change, and time. \cb{This formulation accommodates} complex temporal risk behaviors arising in a range of applications. \cb{Building on this framework, we develop a class of nonlinear and oscillatory hazard models, each associated with an interpretable dynamical mechanism and an induced survival distribution.} We also present a simulation procedure for solving a system of non-linear higher-order ODEs, with failure times generated via cumulative hazard inversion. Likelihood-based Bayesian inference under right censoring is also developed, and moment generating function analysis is used to characterize tail behavior. \mbh{The proposed framework is evaluated through simulation studies and illustrated using real data, demonstrating its ability to capture temporal risk patterns not well represented by standard monotone models. In contrast to existing linear ODE-based hazard models, the proposed approach accommodates nonlinear and non-equilibrium dynamics, enabling the representation of temporal risk patterns that are not well captured by first-order or linear oscillator-based formulations.}

%\cb{demonstrating its ability to capture temporal risk patterns not well represented by standard monotone models.} \mbh{Unlike existing linear ODE-based hazard models, the proposed framework accommodates nonlinear and non-equilibrium dynamics, enabling the representation of temporal risk patterns that cannot be captured by first-order or linear oscillator-based formulations.}
\end{abstract}

%%%%%%%%%%%%%%%%%%%%%%%%%%%%%%%%%%%%%%%%%%

%%%%%%%%%%%%%%%%%%%%%%%%%%%%%%%%%%%%%%%%%%

\section{Introduction}

Survival Analysis is a statistical framework for modeling and analyzing time-to-event data, focusing on the timing of specific events of interest, such as death, system failure, disease progression, or recovery. Classical methods address practical complications such as right censoring and truncation, which arise when complete event information is not available. The core components include the survival function, which estimates the probability of surviving beyond a given time, and the hazard function, which characterizes the instantaneous risk of event occurrence. In particular, the hazard function offers a dynamic perspective on risk evolution over time, which makes it especially suitable for modeling through differential equations.

The historical roots of survival analysis trace back to early life tables, beginning with John Graunt’s pioneering work in 1662 and further advanced by Edmund Halley’s probabilistic life table in 1693 \cite{camilleri2019history}. Traditional survival models typically employ parametric approaches, assuming that survival times follow predefined probability distributions such as Weibull, Exponential, or Log-Normal \cite{taketomi2022parametric}. These methods impose rigid functional forms on survival and hazard functions, leading to limitations in capturing complex survival dynamics \cite{taketomi2022parametric, cordeiro2014exponential}. While effective for hazard functions with constant or monotonic rates, parametric models struggle with time-varying covariates, feedback mechanisms, and model misspecification \cite{cox2007parametric}. In contrast, non-parametric approaches, including the Kaplan-Meier estimator \cite{kaplan1958nonparametric}, which estimates the survival function, and the Nelson-Aalen estimator, which estimates the cumulative hazard function, derive survival characteristics directly from observed data without assuming explicit distributions. While these methods provide data-driven insights, they lack interpretability regarding the underlying mechanisms driving hazard rate changes and are unable to extrapolate beyond the observed time horizon \cite{janurova2014nonparametric,kenah2013non}. Semi-parametric models, such as the Cox proportional hazards model \cite{cox1972regression}, offer a middle ground by allowing covariate effects to be estimated flexibly while leaving the baseline hazard function unspecified. Thus, semi-parametric models are advantageous when the baseline hazard function is unknown but impose other restrictions such as proportional hazards, which may not hold in settings where risk evolves over time. These classical methods thus face limitations in their inability to explicitly model the underlying dynamic mechanisms governing the evolution of the hazard function.

Many real-world survival processes, however, evolve as dynamical systems, where changes in risk depend not only on the current hazard level but also on its rate of change and the surrounding environment. Examples include biological growth and decay processes, disease progression and remission cycles, reliability degradation in engineered systems, and market recovery or failure in economics. Such processes are more naturally described through differential equations than static statistical relationships, motivating the development of dynamical survival analysis. In recent years, ODEs have gained prominence in survival analysis for modeling time-dependent processes. For example, Christen and Rubio \cite{christen2024dynamic} introduced systems of ODEs for hazard modeling, shifting attention from static to dynamic \mbh{representations of risk. In their framework, first-order ODE systems link statistical and mechanistic perspectives by allowing the hazard to evolve over time; however, the rate of change depends only on the current hazard level and time, resulting in inherently Markovian dynamics that limit the ability to represent complex temporal behaviors such as oscillations, delayed feedback, or nonlinear damping \cite{liu2012survival}.} These features are particularly relevant in settings such as remission-relapse disease dynamics, fatigue-driven reliability, and recovery from market or environmental shocks, where risk does not merely rise or fall monotonically.  

%\cb{Recent work \cite{christen2024harmonic} introduced a second-order ODE-based hazard model. Nevertheless, a general framework for systematically studying and estimating higher-order dynamical hazard models remains underdeveloped.}

\mbh{More recently, Christen and Rubio \cite{christen2024harmonic} proposed a second-order hazard model based on the linear damped harmonic oscillator. While this formulation extends the dynamical perspective, it corresponds to a specific parametric model with constant coefficients and closed-form solutions, and its behavior is governed by linear dynamics that converge to equilibrium. In contrast, the framework proposed in this paper is not restricted to a single dynamical system, but instead defines a general class of hazard models governed by nonlinear and potentially non-autonomous higher-order differential equations. This broader formulation allows the hazard dynamics to depend on both the current state and its rate of change through nonlinear interactions, enabling behaviors such as state-dependent oscillations, transient amplification, and non-equilibrium dynamics that cannot be generated by linear oscillator-based models. Thus, our proposed approach can be viewed as a strict generalization of linear oscillator-based hazard models, in which the latter arise as a special case under linear and constant-coefficient dynamics.
}

\mbh{Motivated by these considerations, we develop a broader} \cb{dynamical survival framework by incorporating higher-order ODE representations for the hazard function.} By allowing hazard dynamics to depend on both the current hazard level and its rate of change, the proposed framework captures feedback, inertia, and adaptive responses that are inaccessible to first-order models.  \cb{Within this framework, we develop a class of nonlinear and oscillatory hazard models, including an extension of the damped oscillatory hazard model, a second-order logistic hazard model derived from classical logistic dynamics, as well as sinusoidal and interaction-moderated exponential hazards, each linking an interpretable dynamical mechanism to an induced survival distribution.} \mbh{In addition, we \cb{present} a general numerical framework for solving higher-order hazard systems and simulating event times, which applies uniformly across the models considered. Likelihood-based inference under right censoring is developed, and the long-run behavior of the resulting hazard processes is characterized. Through simulation studies and a real-data application, we demonstrate that higher-order hazard dynamics can capture temporal risk patterns that are not well represented by standard monotone or first-order models. This added flexibility is particularly important in applications where risk evolves through nonlinear feedback, delayed adjustment, or cyclical mechanisms, where} \cb{linear and first-order hazard models, along with the oscillatory hazard model introduced in \cite{christen2024harmonic}, may be inadequate.}

%\mbh{In addition, we \cb{provide} a general numerical framework for solving higher-order hazard systems and simulating event times, applicable uniformly across the models considered. We develop likelihood-based inference under right censoring and characterize the long-run behavior of the resulting hazard processes. Simulation studies and a real-data application illustrate that higher-order hazard dynamics capture temporal risk patterns not well represented by standard monotone or first-order models. This flexibility is particularly important in applications where risk evolves through nonlinear feedback, delayed adjustment, or cyclical mechanisms, for which} \cb{linear and first-order hazard models, along with the oscillatory hazard model introduced in \cite{christen2024harmonic}, may be inadequate.}

The remainder of the paper is organized as follows. We begin by presenting the 
ODE-based modeling framework and briefly reviewing the first-order formulation 
of hazard dynamics. We then extend this framework to hazard functions governed 
by higher-order ODEs in Section~\ref{sec:method}. The probability distribution  induced by ODE-based hazards is discussed in Section~\ref{sec:prob_dist}. 
Section~\ref{sec:non-linear} presents nonlinear and oscillatory hazard models. 
Numerical studies and statistical inference are discussed in Sections~\ref{sec:numerical_study} and \ref{sec:inference}, respectively. 
Finally, Section~\ref{sec:simulation} presents simulation results and a 
real-data application in Section~\ref{sec:real_data}.

%%%%%%%%%%%%%%%%%%%%%%%%%%%%%%%%%%%%%%%%%%
\section{ODE-based Hazard Modeling}\label{sec:method}

\cb{This section outlines the ODE-based framework for modeling hazard dynamics used in this work.} We begin by revisiting the first-order ODE formulation, then \mbh{develop higher-order systems that define a broader class of hazard dynamics.} The approach treats the hazard function and its cumulative form as components of a coupled dynamical system, providing a structured foundation for simulation and inference in later sections. \mbh{In particular, the higher-order formulation allows for nonlinear and state-dependent evolution of the hazard, extending beyond the linear and closed-form models considered in existing ODE-based approaches.}

%%%%%%%%%%%%%%%%%%%%%%%%%%%%%%%%%%%%%%%%%%

\subsection{First-Order ODE Framework for Hazard Dynamics} 

We consider a cohort of $n$ subjects with latent event times $\{o_i\}_{i=1}^n$ and right-censoring times $\{c_i\}_{i=1}^n$. In practice, only the earlier of these two times is observed, defined as \( t_i = \min(o_i, c_i) \), along with a censoring indicator $\delta_i$, given by:

\begin{equation}
    \delta_i = 
\begin{cases}
1, & o_i \leq c_i, \\[6pt]
0, & \text{otherwise}.
\end{cases}
\end{equation}

Each \(\{o_i\}\) is assumed to follow a continuous distribution characterized by the following interrelated quantities: a probability density function \(f(t)\), a cumulative distribution function \(F(t) = \int_0^t f(m) \, \mathrm{d}m\), and a survival function \(S(t) = 1 - F(t)\). The hazard function, which describes the instantaneous rate of event occurrence conditional on survival up to time $t$, is defined as:

\begin{equation}
h(t) = -\frac{S'(t)}{S(t)},
\end{equation}

\noindent and the cumulative hazard function representing the total accumulated risk up to time \(t\) is given by:

\begin{equation}
H(t) = \int_0^t h(m) \, \mathrm{d}m = -\log S(t).
\end{equation}

To model the temporal evolution of the hazard function explicitly within a dynamical framework, we define a system of first-order ODEs. The instantaneous hazard $h(t)$ is embedded as one component of a vector of dynamically evolving state variables,

\begin{equation}
\mathbf{Y}(t) = \big(h(t), q_1(t), \ldots, q_m(t)\big)^\top\quad, t > 0, 
\end{equation}

\noindent where \(q_j: \mathbb{R}^+ \to \mathbb{R}\) (\(j = 1, \ldots, m\)) denotes an auxiliary differentiable function capturing additional latent temporal dynamics in the hazard process. The resulting coupled first-order ODE system governing \(\mathbf{Y}(t)\) and the cumulative hazard \(H(t)\) is then given by:

\begin{equation}
\begin{aligned}
\mathbf{Y}'(t) &= \psi_\theta(\mathbf{Y}(t), t), \\
H'(t) &= Y_1(t),
\end{aligned}
\qquad
\mathbf{Y}(0)=\mathbf{Y}_0,\quad H(0)=0.
\end{equation}

\noindent where \(\psi_\theta: \mathbb{R}^{m+1} \times \mathbb{R}^{+} \rightarrow \mathbb{R}^{m+1}\) is a parameterized vector field capturing the dynamic interactions among hazard function components and auxiliary latent processes. The first component, \(Y_1(t)\), explicitly represents the instantaneous hazard \(h(t)\), ensuring that the cumulative hazard \(H(t)\) is consistently and accurately derived from the system. This formulation naturally captures temporal dependencies, nonlinear behavior, and covariate effects in survival dynamics (see, e.g., \cite{tang2023survival}).

\paragraph{Positivity, Stability, and Structural Constraints:} Ensuring the nonnegativity of \( h(t) \) is essential, since the hazard function is nonnegative by definition. The ODE framework provides several mechanisms for enforcing this property. 
(a) The vector field \(\psi_\theta\) can be constructed so that the hazard dynamics preserve nonnegativity, for example by restricting the domain of the system or imposing suitable structural constraints. 
(b) Alternatively, a transformation such as \( \tilde{h}(t) = \log h(t) \) may be used, which guarantees positivity of \(h(t)\) while allowing unconstrained evolution of \(\tilde{h}(t)\). 
(c) Finally, initializing the system with \( h(0) = h_0 \ge 0 \) and enforcing appropriate stability conditions on \(\psi_\theta\) ensures that nonnegativity is preserved over time.

The inclusion of auxiliary latent states \(q_j(t)\) enables the model to capture time-varying covariate effects and latent dependencies in the hazard process, extending the framework to complex survival dynamics. Simultaneous modeling of \(h(t)\) and its cumulative counterpart \(H(t)\) further ensures internal consistency and facilitates efficient numerical inference. Together, these properties make the ODE framework well suited for modeling complex temporal hazard dynamics, while maintaining key probabilistic constraints \cite{tang2023survival, christen2024dynamic}.

\medskip

\mbh{Despite its flexibility, the first-order formulation remains fundamentally limited in its dynamical structure. Since the evolution of the hazard depends only on its current state and time, the model cannot explicitly capture inertia, acceleration, or delayed feedback. To capture these non-linear dynamics without introducing additional latent states, we need to transition to a higher-order ODE framework for hazard dynamics.}

%%%%%%%%%%%%%%%%%%%%%%%%%%%%%%%%%%%%%%%%
\subsection{Higher-Order ODE Frameworks for Hazard Dynamics}

We now \mbh{develop a higher-order} ODE formulation of hazard dynamics, focusing on the second-order case \mbh{as a fundamental extension of the first-order framework.} Second-order systems introduce an additional degree of freedom that allows the hazard’s rate of change itself to evolve dynamically, \mbh{introducing intrinsic memory and inertia into the system. This results in a qualitatively richer class of models capable of representing feedback-driven and non-monotonic risk evolution.} A graphical abstract of the framework is presented in Figure~\ref{fig:higher_order_ode_survival}.

The hazard function $h(t)$ is modeled as part of a second-order system, with dynamics described by:
\begin{equation}
    h''(t) = \phi_\theta(h(t), h'(t), t), \quad t\geq 0,
\end{equation}   

\noindent where $\phi_\theta : \mathbb{R}^3 \to \mathbb{R}$ is a continuously differentiable function parameterized by $\theta \in \Theta \subseteq \mathbb{R}^p$. The function $\phi_\theta$ specifies how the acceleration of the hazard depends on its current value, its instantaneous rate of change, and time. \mbh{This formulation generalizes linear second-order hazard models by allowing $\phi_\theta$ to be nonlinear and state-dependent. In contrast to linear oscillator-based systems with constant coefficients \cite{christen2024harmonic}, this approach facilitates a broader range of non-equilibrium hazard behaviors.}

To handle second-order dynamics, the system is conveniently reformulated as a first-order system by introducing an auxiliary variable $v(t) = h'(t)$: 

\begin{equation}\label{eq:state_space_second_order}
\begin{cases}
h'(t) = v(t),\\[4pt]
v'(t) = \phi_\theta(h(t), v(t), t).
\end{cases}
\end{equation}
Define the state vector
\[
\mathbf{Y}(t) =
\begin{pmatrix}
h(t)\\[4pt] v(t)
\end{pmatrix},
\qquad
\mathbf{Y}'(t) =
\Psi_\theta(\mathbf{Y}(t), t)
=
\begin{pmatrix}
v(t)\\[4pt] \phi_\theta(h(t), v(t), t)
\end{pmatrix}.
\]
\mbh{which yields a compact state-space representation that facilitates numerical implementation while preserving the richer dynamics of the higher-order formulation.} Second-order systems can naturally capture oscillatory, damped, or feedback-driven behaviors (see, e.g., \cite{rabinovich2012oscillations}), which may be relevant in survival scenarios involving cyclical risks, delayed responses, or adaptive mechanisms. These dynamics enable more realistic representations of complex hazard evolution over time.

\mbh{Importantly, while any second-order system can be rewritten as a first-order system by augmenting the state vector, this transformation alters the interpretation of the model by introducing latent variables rather than explicitly modeling the rate of change of the hazard. By maintaining a higher-order formulation, we preserve the direct interpretability of acceleration and feedback mechanisms—elements that are structurally essential to our proposed modeling approach.}

\begin{figure}[H]
    \centering
    \includegraphics[width=0.7\textwidth]{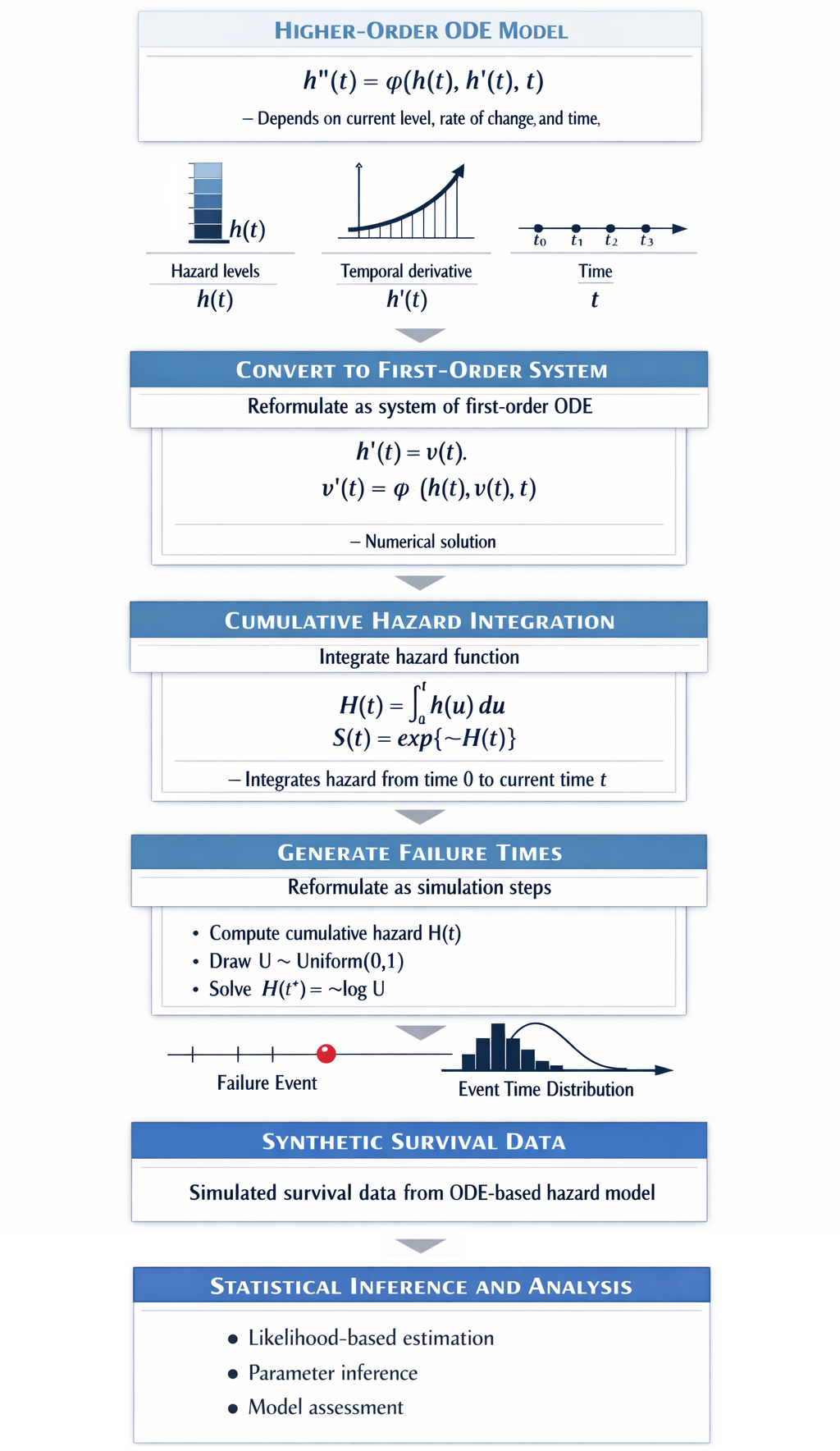}
    \caption{\small{Flowchart illustrating the simulation framework for survival models driven by higher-order ODE-based hazard dynamics. The hazard function is specified via a higher-order ODE and reformulated as a system of first-order equations for numerical solution. The resulting hazard trajectory is integrated to obtain the cumulative hazard function, which is then used to generate event times via inverse transform sampling. This procedure yields synthetic survival data, which can subsequently be used for likelihood-based inference, parameter estimation, and model assessment.}}
    \label{fig:higher_order_ode_survival}
\end{figure}

%%%%%%%%%%%%%%%%%%%%%%%%%%%%%%%%%%%%%%%%%%
\begin{theorem}(Existence, Uniqueness, and Stability of Higher-Order Hazard Dynamics)

\noindent Let the hazard function \( h(t) \) be governed by a second-order ODE of the form:
\[
h''(t) = \phi(h(t), h'(t), t),
\]
where \( \phi: \mathbb{R}^3 \to \mathbb{R} \) is a continuous function, and the initial conditions are specified as \( h(0) = h_0 \) and \( h'(0) = v_0 \). Assume the following:
\begin{enumerate}
    \item \(\phi(h, h', t)\) is Lipschitz continuous in \(h\) and \(h'\), i.e., there exists a constant \(L > 0\) such that:
    \[
    |\phi(h_1, v_1, t) - \phi(h_2, v_2, t)| \leq L(|h_1 - h_2| + |v_1 - v_2|), \quad \forall h_1, h_2, v_1, v_2 \in \mathbb{R}, \, t \geq 0.
    \]
    \item \(\phi(h, h', t)\) is bounded for all \(t \geq 0\), \(h \in \mathbb{R}\), and \(h' \in \mathbb{R}\).
\end{enumerate}
Then, there exists a unique solution \( h(t) \) defined on an interval \([0, T]\) for the given ODE and initial conditions. For stability, if \( \phi(h, h', t) \) is continuously differentiable in \( h \) and \( h' \), and if all eigenvalues of the Jacobian matrix of the system:
    \[
    J = \begin{bmatrix}
    0 & 1 \\
    \dfrac{\partial \phi}{\partial h} & \dfrac{\partial \phi}{\partial h'}
    \end{bmatrix}
    \]

\noindent have negative real parts, then the equilibrium $(h^*,0)$ is locally asymptotically stable, and $h(t)\to h^*$ as $t\to\infty$ for all solutions starting sufficiently close to $(h^*,0)$.
\end{theorem}

%%%%%%%%%%%%%%%%%%%%%%%%%%%%%%%%%%%%%%%%
%%%%%%%%%%%%%%%%%%%%%%%%%%%%%%%%%%%%%%%%
\section{Probability Distribution Induced by First and Second-Order ODEs}\label{sec:prob_dist}

In this section, we demonstrate how conventional survival distributions can be interpreted as particular solutions to first- or higher-order ODE systems. Hazard functions \( h(t) \) can be naturally expressed as solutions to ODEs. In fact, any continuously differentiable hazard function trivially satisfies the first-order ODE:

\[
h'(t) = \psi(h, t),
\]

\noindent Therefore, a wide class of survival models, both parametric and nonparametric, can be viewed as solutions to dynamical systems characterized by suitable choices of $\psi$ or its higher-order analogues \cite{tang2023survival, jiang2020modeling, okagbue2017classes}. \mbh{Within this perspective, the proposed higher-order framework differs in that it explicitly models the dynamical mechanisms governing hazard evolution, rather than treating them as implicit consequences of a predefined functional form.} Below, we explore representations for both first-order and second-order ODEs, highlighting their properties. 

\paragraph{First-Order ODE Representations:}

We now illustrate how standard survival distributions arise as specific solutions of first-order ODEs. For example, consider the Bernoulli-type differential equation:
\begin{equation}\label{eq:hazard_riccati}
h'(t)=a(t)h(t)-h(t)^2,
\end{equation}
\noindent where \( a(t) = f'(t)/f(t) = (\mathrm{d}/\mathrm{d}t) \log f(t) \), and \( f(t) \) is the probability density function (pdf) of the corresponding survival model. This ODE can also be interpreted as a Riccati-type equation \cite{sugie2017nonoscillation}. The coefficient \( a(t) \), referred to as the \textit{autonomy coefficient,} determines whether the ODE is autonomous or nonautonomous: (a) If \( a(t) \) depends only on \( h(t) \), the ODE is autonomous, (b) If \( a(t) \) explicitly depends on \( t \), the ODE is non-autonomous. First-order ODEs with constant or monotonic $a(t)$ often yield hazard functions corresponding to classical monotonic survival distributions. In particular, monotonic hazard functions arise from autonomous scalar ODEs, whereas non-monotonic hazard behavior requires non-autonomous dynamics. For additional discussion on the connection between ODEs and hazard functions in survival models, see \cite{christen2024dynamic}.

\mbh{}

\paragraph{Second-Order ODE Representations and Induced Probability Laws:}

Second-order ODEs provide a richer modeling framework by introducing higher-order dynamics such as acceleration and oscillations. These features are particularly useful for modeling hazard functions with non-monotonic behavior, including periodic or fluctuating risks. A second-order hazard model is given by

\[
h''(t)=\phi(h(t),h'(t),t),
\]
where \(\phi\) specifies the dependence of the hazard acceleration on its current level, rate of change, and time. Rewriting this system in state-space form,
\[
z(t)=h'(t), \qquad z'(t)=\phi(h(t),z(t),t),
\]
where \(z(t)\) denotes the instantaneous rate of change of the hazard.

For each specified hazard model, the corresponding probability density function is given by
\begin{equation*}\label{eq:pdf.ht}
f(t) = h(t)\exp\{-H(t)\}.
\end{equation*}

Random realizations can be generated using inverse transform sampling. Specifically, for $u \sim \mathrm{Uniform}(0,1)$, event times $t$ are obtained by solving $H(t) = -\log(1-u).$

In the next section, we specify concrete examples of higher-order hazard systems, each illustrating a distinct nonlinear or oscillatory dynamic together with its induced survival and density functions, thereby linking the ODE formulation to empirically relevant survival behavior.

%%%%%%%%%%%%%%%%%%%%%%%%%%%%%%%%%%%%%%%%
%%%%%%%%%%%%%%%%%%%%%%%%%%%%%%%%%%%%%%%%

\section{Nonlinear and Oscillatory Hazard Models via Higher-Order ODEs}\label{sec:non-linear}

In this section, \mbh{we illustrate the proposed higher-order dynamical framework through four representative second-order ODE-based hazard models. While the general formulation accommodates a broad class of nonlinear systems, we focus on models that arise naturally in time-to-event applications and highlight distinct dynamical mechanisms.} The models examined are: (1) a population dynamics-based hazard model, (2) an exponential hazard model with interaction effects, (3) a damped oscillatory hazard model, and (4) a sinusoidal hazard model. \mbh{Collectively, these examples demonstrate how the proposed framework captures diverse behaviors-including damping, oscillations, nonlinear growth and saturation, and interaction-driven dynamics-that are not accurately represented by standard first-order or linear oscillator-based hazard models.}

%The models examined are: (1) a damped oscillatory hazard model, (2) a population dynamics–based hazard model, (3) a sinusoidal hazard model, and (4) an exponential hazard model with interaction effects. \mbh{Collectively, these examples demonstrate how the proposed framework captures diverse behaviors—including damping, oscillations, nonlinear growth and saturation, and interaction-driven dynamics—that are not accurately represented by standard first-order or linear oscillator-based hazard models.}
%%%%%%%%%%%%%%%%%%%%%%%%%%%%%%%%%%%%%%%%%%
\subsection{Population Dynamics-Based Hazard Model}  
The classical logistic growth law is widely used in ecology to describe population growth under limited resources\cite{Pielou1977,murray2007mathematical}.
In its standard first–order form, the dynamics are
\begin{equation}
\label{eq:firstorderlogistic}
    h^\prime(t)
    \;=\;
    r\,h(t)\!\left(1 - \frac{h(t)}{K}\right),
\end{equation}
where $r>0$ is the intrinsic growth rate and $K>0$ is the carrying capacity.
The analytic solution is the familiar sigmoidal curve
\begin{equation*}
\label{eq:firstorder_solution}
    h(t)
    \;=\;
    \frac{K\,h_0\,e^{rt}}
         {K + h_0\!\left(e^{rt}-1\right)},
\end{equation*}
showing that in the first–order formulation the trajectory of $h(t)$ is determined entirely by its current value \cite{Verhulst1838}. Because
Eq.~\eqref{eq:firstorderlogistic} specifies only the instantaneous rate of change, the dynamics react
immediately to the current state and cannot capture inertia or acceleration. \cb{Consequently, the model produces only monotone
convergence to equilibrium and cannot generate overshoot, oscillations, or transient amplification.} \mbh{In contrast to linear second-order hazard models, which impose a fixed restoring structure, the logistic formulation introduces nonlinear state-dependent feedback, allowing the dynamics to vary with the level of the hazard itself.}

To incorporate inertia and state–dependent acceleration, we propose the nonlinear
second–order logistic system
\begin{equation}
\label{eq:popdyn_main}
    h^{\prime \prime}(t) = \phi(h(t),h^\prime(t),t) =
    r\,h(t)\!\left(1 - \frac{h(t)}{K}\right).
\end{equation}

\noindent Here the logistic feedback acts on the acceleration of the hazard. Unlike the harmonic–oscillator hazard model of \cite{ChristenRubio2025HOH}, which is a linear second-order ODE with a closed form solution, Eq.~\eqref{eq:popdyn_main} is \emph{nonlinear} and autonomous. Transient dynamics (overshoot, oscillations, non–monotone adjustment toward $K$) arise intrinsically from the logistic mechanism itself, rather
than from a linear restoring force.

Since Eq.~\eqref{eq:popdyn_main} has no simple closed–form solution, we introduce $v(t) = h^\prime(t)$ and rewrite it as a first–order system
\begin{equation}
\label{eq:popdyn_system}
\begin{cases}
h^\prime(t) \;=\; v(t),\\[4pt]
v'(t) \;=\; r\,h(t)\!\left(1 - \dfrac{h(t)}{K}\right),
\end{cases}
\qquad
h(0)=h_0,\;\; v(0)=v_0 .
\end{equation}

For numerical simulation we include a small damping term
$\eta\,h^\prime(t)$, $\eta>0$, giving
\[
h^{\prime \prime}(t) + \eta\,h^\prime(t)
\;=\;
r\,h(t)\left(1 - \frac{h(t)}{K}\right),
\]

which suppresses sustained oscillations and stabilizes convergence to $K$. Equivalently, this can be written as
\begin{equation}
\label{eq:pd_main}
\begin{cases}
h'(t) = v(t), \\
v'(t) = r\,h(t)\!\left(1-\dfrac{h(t)}{K}\right) - \eta\,v(t),
\end{cases}
\qquad h(0)=h_0,\; v(0)=v_0 .
\end{equation}

For sufficiently large damping,solutions initiated with $h_0>0$ remain positive and converge to $K$ over the range of initial conditions considered. Accordingly, the simulation study restricts attention to $h_0>0$ and $\eta>0$ for which the numerical solution satisfies $h(t)\ge 0$ for all $t$.

To better understand the dynamical role of the parameters, we rescale
\[
x(t) = \frac{h(t)}{K},
\qquad
\tau = \sqrt{r}\,t ,
\]
which converts the damped second–order logistic equation into
\[
\frac{d^2 x}{d\tau^2} + \zeta\,\frac{dx}{d\tau}
\;=\;
x(1-x),
\qquad
\zeta = \frac{\eta}{\sqrt{r}}.
\]
Thus, the transient behavior is governed by the damping parameter $\zeta$.  Large damping
($\zeta \gg 1$) produces monotone convergence, whereas weak damping
($\zeta \ll 1$) produces overshoot and damped oscillatory adjustment toward $x=1$ (i.e.,
$h(t)=K$). Damping controls the decay rate of oscillations, while the inertia generates the initial overshoot.

Classical extensions of the logistic law generate non–monotone trajectories by
adding \emph{memory}. A delayed logistic model,
\begin{equation}
\label{eq:hazard_delay}
    h^\prime(t)
    \;=\;
    r\,h(t)\!\left(1 - \frac{h(t-\tau)}{K}\right),
\end{equation}
produces overshoot or oscillation when the delay $\tau$ is sufficiently large,
because the system reacts to the past state $h(t-\tau)$
\cite{Rosen1987,ZhangGopalsamy1990,SchleyGourley2000}.
Here, complex behavior is driven by history. In contrast,Eq.~\eqref{eq:popdyn_main} produces similar transient effects
(overshoot, oscillation) \emph{without requiring delay}: inertia replaces memory \mbh{as the driving mechanism. This highlights a key distinction between our proposed framework and delay-based models, with higher-order dynamics providing an intrinsic mechanism for complex temporal behavior.}

Figure~\ref{fig:logistic_comparison} illustrates these differences: the first–order model converges monotonically, the delayed model overshoots due to feedback on a past state, and the second–order model overshoots because acceleration carries the system beyond $K$ before it slows down.

\cb{
Oscillatory hazard patterns can also be modeled using damped oscillatory hazard 
formulations (Section \ref{sec:DO}), which provide a flexible statistical description 
of non-monotone hazard behavior. In contrast, the second-order logistic model generates
oscillations through the underlying system dynamics. Near equilibrium, these dynamics 
reduce to a damped harmonic oscillator, implying that the oscillatory hazard model can 
be viewed as a local linear approximation of the logistic system. Accordingly, the two models 
play complementary roles: the oscillatory model provides a descriptive benchmark, while the logistic formulation offers a mechanistic explanation of transient oscillatory hazards.}
\begin{figure}[H]
\centering
\includegraphics[scale=0.7]{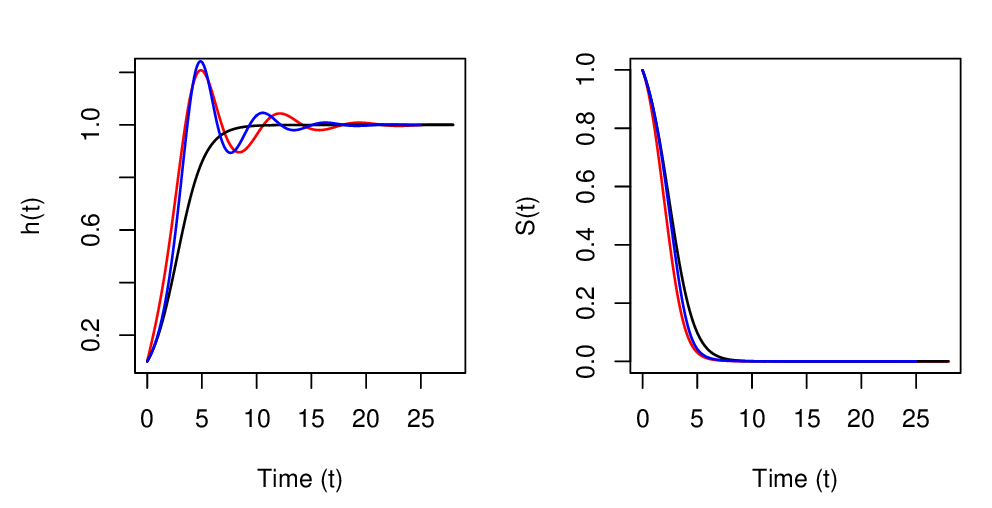}
\caption{
Hazard functions $h(t)$ and survival functions $S(t)$ for logistic hazard dynamics between first-order (black), delayed (blue), and damped second-order logistic (red) hazard models
($r=0.8,\; K=1,\; \tau=1.2,\; \zeta=0.5,\; h_0=0.1,\; v_0=0.2$).}
\label{fig:logistic_comparison}
\end{figure}

%%%%%%%%%%%%%%%%%%%%%%%%%%%%%%%%%%%%%%%%%%
%%%%%%%%%%%%%%%%%%%%%%%%%%%%%%%%%%%%%%%%%%
\subsection{Exponential Hazard Model with Interaction}
 % removing decay
\mbh{The exponential function plays a central role in modeling dynamic processes in physics and biology, including population growth, tumor development, and pharmacokinetics. However, pure exponential growth is often unrealistic, as it does not account for constraints such as limited resources or environmental capacity. To address this limitation, we introduce an exponential hazard model with interaction that extends the basic exponential framework by incorporating nonlinear coupling effects. Within the proposed higher-order formulation, this model illustrates how interactions between the hazard level and its rate of change can generate dynamics that are not attainable in linear second-order systems. In this section, we focus on the case of exponential growth.}

We define an exponential hazard model with interaction through the second-order ODE
\begin{equation}\label{eq:expo_grow}
h''(t)=\phi\!\big(h(t),h'(t),t\big)=\alpha h(t)-\beta\,\big(h'(t)\big)^2 .
\end{equation}
\noindent \mbh{The quadratic term $(h'(t))^2$ introduces nonlinear coupling into the system, making the acceleration depend on the magnitude of the hazard’s rate of change. This contrasts with linear oscillator-based models, where the dynamics depend only linearly on $h(t)$ and $h'(t)$. As a result, the interaction moderates the growth of the hazard relative to the case $\beta=0$, with stronger damping effects when the rate of change is large.}

For analysis and computation, Eq.~\eqref{eq:expo_grow} may be written as the equivalent first-order system
\begin{equation}
    \begin{cases}
h'(t) = v(t), \\
v'(t) = \alpha h(t) - \beta v(t)^2
\end{cases}
\qquad
h(0)=h_0,\;\; v(0)=v_0 .
\end{equation}

The model is parameterized by $(h_0, v_0, \alpha, \beta)$, where $h_0 = h(0)$ and
$v_0 = h'(0)$ specify the initial condition. The parameter $\alpha$ governs the
growth component of the hazard, while $\beta$ introduces a nonlinear interaction
through the rate of change $h'(t)$. \mbh{The associated quadratic term acts as a state-dependent damping mechanism that becomes more pronounced when $h'(t)$ is large, leading to adaptive moderation of hazard growth.}

\vspace{-6pt}
\paragraph{Linear Case:} When $\beta = 0$, the interaction term vanishes and the model reduces to the linear second-order equation
\[
h''(t) = \alpha h(t).
\]
For $\alpha > 0$, the solution consists of exponentially growing and decaying components, with long-term behavior dominated by exponential growth. This case provides a useful baseline for comparison with the nonlinear model, highlighting how the interaction term alters the hazard dynamics. Closed-form expressions for the solution and cumulative hazard are provided in Appendix~\ref{app:exponential}.

\vspace{-6pt}
\paragraph{General Case:}
When $\beta>0$, the nonlinear interaction term in Eq.~\eqref{eq:expo_grow} \mbh{fundamentally alters the system dynamics. The quadratic dependence on $h'(t)$ introduces state-dependent damping, which moderates the growth of the hazard as its rate of change increases. As a result, the model captures adaptive growth behavior that deviates substantially from pure exponential trajectories, a feature not attainable within linear second-order systems.} In this case, no closed-form analytic solution is available, and the solution must be obtained numerically for given initial conditions. Numerical solutions for hazard models without closed-form expressions can be obtained using standard ODE solvers, such as Runge-Kutta methods (see, e.g., \cite{zingg1999runge, soetaert2010solving}),
with specified initial conditions and sufficiently small step sizes to ensure numerical stability and accuracy. Figure~\ref{fig:plot-exp} illustrates the behavior of the exponential hazard model for $\beta=0$ and $\beta>0$. Introducing the interaction term moderates the growth of the hazard, resulting in a slower increase in $h(t)$ compared with the case $\beta=0$. This change in hazard dynamics is also reflected in the survival functions, where the interaction affects both the timing and the rate of survival decay.

%%%%%%%%%%%%%%%%%%%%%%%%%%%%%%%%%%%%%%%%%%
\begin{figure}[H]
    \centering
    \includegraphics[scale=0.7]{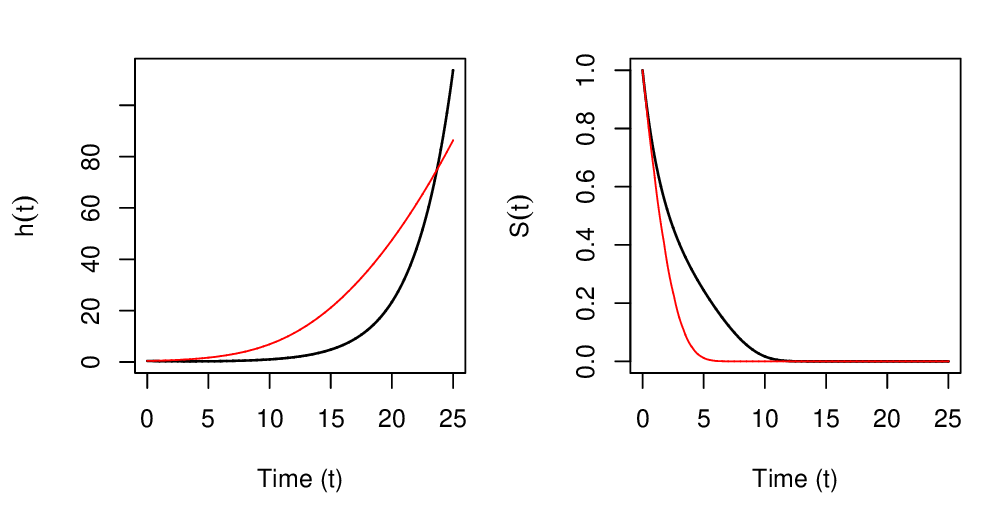}
    \caption{Hazard function $h(t)$ and survival function $S(t)$ for the exponential hazard model without interaction ($\beta=0$, black) and with nonlinear interaction ($\beta=0.1$, red), with parameters $\alpha = 0.1,\ h_0 = 0.4,\ v_0 = 0.1$.}
    \label{fig:plot-exp}
\end{figure}

%%%%%%%%%%%%%%%%%%%%%%%%%%%%%%%%%%%%%%%%%%

\subsection{Damped Oscillatory Hazard Model}\label{sec:DO}
The Damped Oscillatory Hazard Model describes hazard behavior that fluctuates over time with gradually decreasing intensity (see e.g. in \cite{socol2020damped, price1991insights}). This class of models is particularly useful in settings where the risk does not evolve monotonically but exhibits oscillatory behavior before stabilizing, for example due to adaptation, periodic external factors, or learning before eventually stabilizing (see e.g. \cite{gitterman2012noisy}). By capturing both cyclical nature and damping effects, this model provides a more realistic and flexible framework to understand time-varying risks in complex systems.

%\cb{The damped oscillatory hazard model has recently been proposed in the literature \cite{christen2024harmonic}. In this work, we express the model within a general second-order ODE framework and rewrite it as an equivalent first-order state-space system. The parameterization $(\alpha,\beta,\gamma)$ represents the model directly through the coefficients of the differential equation, which simplifies parameter estimation and interpretation in survival modeling.}

\mbh{The damped oscillatory hazard model has recently been studied by Christen and Rubio \cite{christen2024harmonic} as a specific linear second-order system with constant coefficients. Within the present framework, this model arises as a particular instance of the general higher-order formulation. Here, we express the model in a unified second-order ODE representation and embed it within a broader class of dynamical systems, which facilitates comparison, extension, and generalization. In contrast to the original formulation, the parameterization $(\alpha,\beta,\gamma)$ directly corresponds to the coefficients of the governing differential equation. This representation simplifies interpretation and estimation, and provides a natural foundation for extending the model to nonlinear and state-dependent dynamics considered later in this section.}

\cb{The damped oscillatory hazard model can be written in the general second-order form
\[
h''(t) = \phi(h(t), h'(t), t),
\]
where
\[
\phi(h(t), h'(t), t) = -\alpha h'(t) - \beta h(t) + \gamma.
\]

\noindent \mbh{In this case, the function $\phi$ is linear in both $h(t)$ and $h'(t)$, corresponding to a constant-coefficient system. While this yields closed-form solutions and tractable analysis, it represents only a restricted subset of the broader class of nonlinear systems encompassed by the proposed framework.}

Introducing the auxiliary variable \(z(t) = h'(t)\), the model can be expressed as the equivalent first-order state-space system
\begin{equation}
\label{eq:osc.ht}
 h'(t) = z(t), \qquad
z'(t) = -\alpha z(t) - \beta h(t) + \gamma.   
\end{equation}

}

\noindent \mbh{The qualitative behavior of the system depends on the discriminant $\Delta = \alpha^2 - 4\beta$ of the characteristic equation $r^{2}+\alpha r+\beta=0$, leading to three distinct cases: underdamped $(\Delta<0)$, critically damped $(\Delta=0)$, and overdamped $(\Delta>0)$. In the underdamped regime, the hazard exhibits oscillatory behavior with exponentially decaying amplitude, while in the critically damped and overdamped regimes, the hazard converges monotonically to its equilibrium level. Closed-form expressions for these solutions, along with the corresponding cumulative hazard functions, are provided in Appendix~\ref{app:closeform_damped}. 

The cumulative hazard function $H(t)$ is obtained by integrating $h(t)$, and in the linear oscillatory case this can be evaluated in closed form. These expressions, together with the corresponding density $f(t) = h(t)\exp{-H(t)}$, enable direct simulation of event times via inverse transform sampling. For the more general nonlinear systems considered later, $H(t)$ is evaluated numerically, and event times are generated accordingly.}

% \noindent The form of the solution is determined by the discriminant $\Delta=\alpha^{2}-4\beta$ of the characteristic equation $r^{2}+\alpha r+\beta=0$, leading to three distinct cases: underdamped $(\Delta<0)$, critically damped $(\Delta=0)$, and overdamped $(\Delta>0)$.

In Figure~\ref{fig:oscillatory}, we illustrate the typical behavior of the three damping cases. The underdamped case shows oscillations in the hazard before settling, the critically damped case returns to equilibrium smoothly without oscillation, and the overdamped case approaches equilibrium more slowly. These differences produce small early-time variations in the survival curves, while the long-run behavior is similar across all three cases.

\begin{figure}[h!]
    \centering
    \includegraphics[scale=0.7]{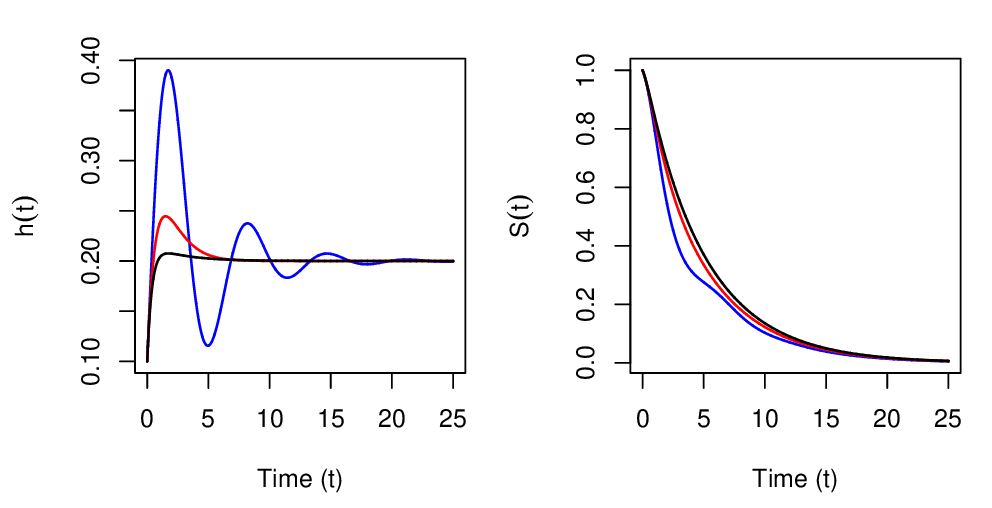}
    \caption{Hazard functions $h(t)$ and survival functions $S(t)$ for the damped oscillatory model under the underdamped ($\alpha = 0.5, \ \beta = 1, \ \gamma = 0.2, \ h_0 = 0.1, \ v_0 = 0.3$; blue),  critically damped ($\alpha = 2, \ \beta = 1, \ \gamma = 0.2, \ h_0 = 0.1, \ v_0 = 0.3$; red), and overdamped ($\alpha = 3, \ \beta = 1, \ \gamma = 0.2, \ h_0 = 0.1, \ v_0 = 0.3$; black) cases.}
    \label{fig:oscillatory}
\end{figure}

\begin{theorem}\label{thm:osc.longterm}(Long-Term Behavior of the Damped Oscillatory Hazard)

Consider the damped oscillator Eq.~\eqref{eq:osc.ht}, with initial conditions \( h(0) = h_0 >0 \) and \( h'(0) = v_0 \). Assume the unique solution satisfies $h(t) >0$ for all $t \geq 0$. Then:

\begin{enumerate}
    \item \textbf{Existence and Uniqueness:} For any $h_0>0$ and $v_0 \in R$, there exists a unique solution $h(t)$ defined for all $t \geq 0$.
    
    \item \textbf{Asymptotic Stability:} For any $\alpha>0$ and $\beta>0$, the equilibrium $h^*=\gamma/\beta$ is globally asymptotically stable, and
\[
\lim_{t\to\infty} h(t)=h^*=\frac{\gamma}{\beta}.
\]
    \item \textbf{Damping Regimes:} Let $\Delta=\alpha^2-4\beta$. 
    \begin{itemize}
    \item If $\Delta<0$, the solution exhibits oscillations about $h^*$ with exponentially decaying amplitude (underdamped).
\item If $\Delta=0$, the solution converges to $h^*$ at the critical damping rate.
\item If $\Delta>0$, the solution converges monotonically to $h^*$ as a sum of decaying exponentials (overdamped).
\end{itemize}
\end{enumerate}
\end{theorem}

\begin{proof}
The proof is given in Appendix~\ref{app:osc-proof}.
\end{proof}

\paragraph{Relationship with the Oscillatory Hazard Model:}

\mbh{We now examine the relationship between the proposed formulation and the oscillatory hazard model of Christen and Rubio \cite{christen2024harmonic}. While the two models are mathematically equivalent under a suitable parameter transformation, their parameterizations differ substantially in structure and interpretability, with important implications for statistical inference.} The oscillatory hazard model in \cite{christen2024harmonic}  can be written as
\[
h''(t) + 2\eta\omega_0 h'(t) + \omega_0^2\bigl(h(t)-h_b\bigr) = 0,
\]
which is equivalent to
\[
h''(t) + 2\eta\omega_0 h'(t) + \omega_0^2 h(t) = \omega_0^2 h_b.
\]

Both this formulation and the model in Eq.~\eqref{eq:osc.ht} can be expressed in the unified form
\[
h''(t)+a_1 h'(t)+a_2 h(t)=a_3,
\]
where \(\mathbf a=(a_1,a_2,a_3)^\top\) denotes the vector of ODE coefficients. For the model in Eq.~\eqref{eq:osc.ht}, the parameter-to-coefficient map is
$T(\alpha,\beta,\gamma)=(\alpha,\beta,\gamma),$
so that each parameter directly corresponds to one coefficient of the differential equation. For the Christen--Rubio parameterization the corresponding map is $T_{\mathrm{CR}}(\eta,\omega_0,h_b)
=
\bigl(2\eta\omega_0,\ \omega_0^2,\ \omega_0^2 h_b\bigr).$ Thus the two formulations are equivalent under the parameter transformation
\[
\alpha = 2\eta\omega_0,
\qquad
\beta = \omega_0^2,
\qquad
\gamma = \omega_0^2 h_b.
\]

To examine the local structure of the two parameterizations, consider the Jacobian matrices of the coefficient maps.  
For the proposed parameterization,
\[
J
=
\dfrac{\partial(a_1,a_2,a_3)}{\partial(\alpha,\beta,\gamma)}
=
\begin{pmatrix}
1 & 0 & 0\\
0 & 1 & 0\\
0 & 0 & 1
\end{pmatrix},
\]
so that
\[
J^\top J = I_3,
\]
indicating that the parameters align directly with the coefficients of the differential equation. For the Christen--Rubio parameterization the Jacobian matrix is
\[
J_{\mathrm{CR}}
=
\dfrac{\partial(a_1,a_2,a_3)}{\partial(\eta,\omega_0,h_b)}
=
\begin{pmatrix}
2\omega_0 & 2\eta & 0\\
0 & 2\omega_0 & 0\\
0 & 2\omega_0 h_b & \omega_0^2
\end{pmatrix}.
\]

This representation shows that perturbations in $\omega_0$ simultaneously affect all coefficients of the differential equation, whereas perturbations in $\eta$ and $h_b$ influence the system through their interaction with $\omega_0$, resulting in locally coupled parameter effects. Consequently, the induced parameter geometry exhibits local coupling among parameters. Indeed,
\[
J_{\mathrm{CR}}^\top J_{\mathrm{CR}} =
\begin{pmatrix}
4\omega_0^2 & 4\eta\omega_0 & 0\\
4\eta\omega_0 &
4\eta^2 + 4\omega_0^2 + 4\omega_0^2 h_b^2 &
2\omega_0^3 h_b\\
0 & 2\omega_0^3 h_b & \omega_0^4
\end{pmatrix},
\]
which contains nonzero off–diagonal terms indicating local parameter coupling. Although the two formulations define the same hazard family, the choice of parameterization can influence prior specification in Bayesian analysis. Under the transformation $(\alpha,\beta,\gamma) = (2\eta\omega_0,\ \omega_0^2,\ \omega_0^2 h_b),$
the determinant of the Jacobian is
\[
\det(J_{CR})=4\omega_0^4.
\]
Therefore, independent priors assigned to $(\eta,\omega_0,h_b)$ induce the following prior on the ODE coefficients:
\[
\pi_{\alpha,\beta,\gamma}(\alpha,\beta,\gamma)
=
\pi_{\eta,\omega_0,h_b}
\!\left(
\frac{\alpha}{2\sqrt{\beta}},
\sqrt{\beta},
\frac{\gamma}{\beta}
\right)
\frac{1}{4\beta^2},
\qquad \beta>0.
\]

Thus, even when $(\eta,\omega_0,h_b)$ are assigned independent priors, the implied prior on $(\alpha,\beta,\gamma)$ becomes nonlinear and dependent. This is a general consequence of nonlinear parameter transformations. \mbh{In contrast, our proposed parameterization allows priors to be specified directly on the coefficients of the governing differential equation, leading to a more transparent interpretation and a simpler local geometry for inference. This distinction is not merely algebraic, but has practical implications for model specification, identifiability, and the interpretation of dynamical effects in hazard modeling.}

\subsection{Sinusoidal Hazard Model}

Standard hazard models, such as the exponential and Weibull, assume constant or monotonic hazard functions \cite{stanley2016comparison}. However, in many real-world situations, the hazard may vary periodically due to factors like seasonal changes or biological rhythms \cite{helm2013annual, aschoff1967adaptive, hridoy2025exploration}. In such cases, sinusoidal hazard models provide a natural alternative by capturing these recurring risk patterns over time. This model is particularly well suited for contexts involving periodic or cyclical hazard behavior, including seasonal effects on survival (e.g., weather-related risk variations), recurring health events (e.g., annual epidemics), and cyclical influences from market or environmental conditions. Unlike the damped oscillatory hazard model, which exhibits decaying fluctuations over time due to damping, the sinusoidal model preserves constant amplitude and frequency throughout the time domain. \mbh{Within the proposed higher-order framework, such periodic hazard behavior arises naturally from second-order dynamics, without requiring external forcing or time-dependent coefficients.}

The sinusoidal hazard model is defined by the second-order ODE 
\[
h''(t) = \phi(h(t), h^\prime(t),t) = -\omega^2  h(t)
\]

\noindent whose general solution is $h(t) = A \cos(\omega t) + B \sin(\omega t)$.

\noindent Since this form oscillates around zero and may take negative values, we introduce a constant shift $(c)$ to ensure $h(t)>0$ for all $t$. We therefore consider the modified equation

\begin{equation}
h''(t) = -\omega^2\big(h(t)-c\big).
\end{equation}

\noindent with solution

\begin{equation}
h(t) = A \cos(\omega t) + B \sin(\omega t) + c 
\end{equation}

\mbh{This formulation corresponds to a linear second-order system with purely oscillatory dynamics and no damping, resulting in sustained periodic behavior. Unlike damped oscillatory hazard models, which converge to equilibrium, this system does not admit a stable steady state and instead produces persistent fluctuations over time.} The parameter $\omega$ controls the oscillation frequency, and $c$ is the baseline hazard level, capturing the persistent risk level that remains constant regardless of the cyclical fluctuations. The oscillatory component has amplitude $R = \sqrt{A^2 + B^2}$, so the hazard fluctuates around $c$ with amplitude $R$. Given initial conditions $h(0)=h_0$ and $h'(0)=v_0$, the solution can be written as
\begin{equation}
\label{eq:ht.sinusoidal}
h(t) = (h_0 - c)\cos(\omega t) + \frac{v_0}{\omega}\sin(\omega t) + c ,
\end{equation}

\noindent in which case the amplitude reduces to
\[
R=\sqrt{(h_0-c)^2+\left(\frac{v_0}{\omega}\right)^2},
\]
and the minimum of the hazard over time is therefore $\min_t h(t)=c-R$.  For $h_0>0$, strict positivity of $h(t)$ for all $t$ is equivalent to
\[
c>\frac{h_0}{2}+\frac{v_0^2}{2h_0\omega^2}.
\]

\noindent \mbh{This ensures that the hazard remains strictly positive while preserving the periodic structure of the model.
}The corresponding cumulative hazard function follows from integration of $h(t)$ and it is given by
\begin{equation*}
H(t) = ct + \frac{h_0 - c}{\omega} \sin(\omega t) + \frac{v_0}{\omega^2} \{ 1- \cos(\omega t)\}.
\end{equation*}

\noindent The pdf is
\begin{equation*}
f(t) = \left( c + (h_0-c) \cos(\omega t) + \frac{v_0}{\omega} \sin(\omega t) \right) 
\exp\left(- ct - \frac{h_0-c}{\omega} \sin(\omega t) - \frac{v_0}{\omega^2} ( 1- \cos(\omega t)) \right).
\end{equation*}

Figure~\ref{fig:plot-sinusodial} illustrates the typical behavior of the sinusoidal hazard model. In contrast to the oscillatory hazard model, which eventually converges to a steady level, the sinusoidal hazard model continues to fluctuate over time and does not approach equilibrium. \mbh{This makes the model particularly suitable for settings where risk follows recurring patterns rather than stabilizing over time, such as seasonal disease incidence or periodic environmental exposures. More broadly, it highlights the ability of the proposed higher-order framework to represent both transient and sustained oscillatory behavior, depending on the underlying dynamical structure.}

%%%%%%%%%%%%%%%%%%%%%%%%%%%%%%%%%%%%%%%%%%
\begin{figure}[H]
    \centering
    \includegraphics[scale=0.7]{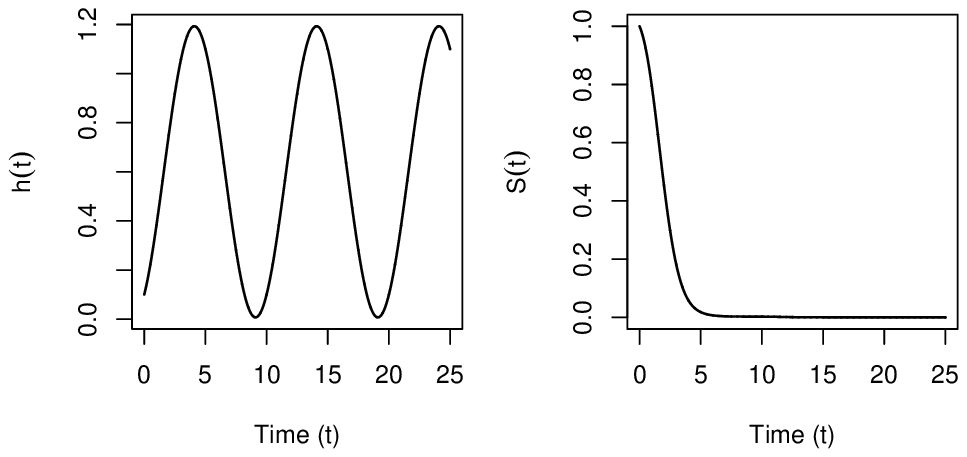}
    \caption{Hazard function $h(t)$ and survival function $S(t)$ for the sinusoidal hazard model ($h_0 = 0.1,\ v_0 = 0.2, \ \omega = 0.2\pi, \ c = 0.6$).}
    \label{fig:plot-sinusodial}
\end{figure}

%%%%%%%%%%%%%%%%%%%%%%%%%%%%%%%%%%%%%%%%%%
%%%%%%%%%%%%%%%%%%%%%%%%%%%%%%%%%%%%%%%%%%

\section{Numerical Algorithm for Higher Order ODE Hazard Models}\label{sec:numerical_study}

In this section, we describe a general numerical algorithm applicable to all proposed ODE-based hazard models. This algorithm provides a general framework for solving the corresponding systems of second-order ODEs, and generating survival time realizations. \mbh{It accommodates both linear and nonlinear systems and can be implemented using standard numerical solvers such as the Runge--Kutta method.}

\mbh{This procedure provides a unified simulation framework for all models considered in this paper, including both linear systems with closed-form solutions and nonlinear systems requiring numerical integration. In particular, it enables the generation of survival data from higher-order hazard dynamics without relying on analytical tractability.}

\begin{algorithm}[H]
\caption{Numerical Simulation of Second-Order ODE: $h''(t) = \phi_\theta(h(t), h'(t), t)$}
\begin{algorithmic}[1]
\REQUIRE Model parameters $\theta$, initial values $h(t_0)$ and $v(t_0) = h'(t_0)$, time interval $[t_0, T]$, step size $\Delta t$
\ENSURE Discrete trajectories of $h(t)$, $v(t)$, cumulative hazard $H(t)$, and simulated failure time $t^*$

\STATE Define uniform time grid: $t_i = t_0 + i \cdot \Delta t$ for $i = 0, 1, \dots, N$ such that $t_N = T$
\STATE Initialize:
\[
h_0 \gets h(t_0), \quad v_0 \gets v(t_0), \quad H_0 \gets 0
\]
\FOR{$i = 1$ to $N$} \STATE Update $(h(t), v(t))$ from $t_{i-1}$ to $t_i$ using a 4th-order Runge--Kutta scheme applied to the system
\[
\begin{cases}
\dfrac{dh}{dt} = v(t), \
\dfrac{dv}{dt} = \phi_\theta(h(t), v(t), t)
\end{cases}
\]

    \begin{enumerate}
        \item Compute RK4 steps $k_1, k_2, k_3, k_4$ for both $h$ and $v$
        \item Update:
        \[
        h_i \gets h_{i-1} + \frac{\Delta t}{6}(k_{1h} + 2k_{2h} + 2k_{3h} + k_{4h})
        \]
        \[
        v_i \gets v_{i-1} + \frac{\Delta t}{6}(k_{1v} + 2k_{2v} + 2k_{3v} + k_{4v})
        \]
    \end{enumerate}
    \STATE Compute cumulative hazard at $t_i$ using trapezoidal rule:
    \[
    H_i \gets H_{i-1} + \frac{1}{2} \left( h_{i-1} + h_i \right) \Delta t
    \]
\ENDFOR

\STATE \textbf{Simulate failure time:}
\STATE Draw $u \sim \text{Uniform}(0,1)$
\STATE Solve for $t^*$ such that:
\[
H(t^*) = -\log(1 - u)
\]
using root-finding (e.g., bisection or \texttt{uniroot()} in R), based on interpolating the discrete cumulative hazard trajectory $\{(t_i, H_i)\}$

\RETURN Arrays $\{h_i\}_{i=0}^N$, $\{v_i\}_{i=0}^N$, $\{H_i\}_{i=0}^N$, and simulated failure time $t^*$
\end{algorithmic}
\end{algorithm}

%%%%%%%%%%%%%%%%%%%%%%%%%%%%%%%%%%%%%%%%%%
%%%%%%%%%%%%%%%%%%%%%%%%%%%%%%%%%%%%%%%%%%

\section{Inference}\label{sec:inference}

Moments and moment generating functions (MGFs) are commonly used in
survival analysis to describe the distribution of failure times and to
support inference on quantities such as the mean survival time and
variability \cite{chakraborti2019higher, song2019moment}. In
degradation--threshold models, MGFs are also used to study first-passage
time distributions, which determine the associated hazard and survival
functions. In addition, parameter estimation for the proposed hazard
models is carried out using maximum likelihood estimation (MLE), which
provides a principled framework for inference based on observed event
times and censoring. In this section, we focus on moment generating
functions and likelihood-based inference for proposed hazard models defined through higher-order ordinary differential equations.

\subsection{Moment Generating Function}

Let $T$ denote a non-negative random variable with density $f(t)$ defined in Eq.~\eqref{eq:pdf.ht}. The MGF of $T$ is given by

\[
M_T(s) = \mathbb{E}[e^{sT}] = \int_0^\infty e^{s t} f(t) \, dt
\]
\noindent for all $s$ in the set for which the integral exists.

For the hazard models considered in this work, the existence and domain
of the MGF depend on the asymptotic behavior of $H(t)$. Since $h(t)\ge 0$,
$H(t)$ is monotone increasing and its asymptotic growth rate determines the tail behavior and the existence of the MGF. Closed-form expressions for $M_T(s)$ are generally unavailable for
the proposed hazard models, and numerical integration is therefore
required. Nevertheless, the existence of the MGF can be characterized
for each model based on the asymptotic behavior of $H(t)$, as discussed
below.

For the damped oscillatory hazard model, Theorem ~\ref{thm:osc.longterm} shows that
$h(t)\to \gamma/\beta$ as $t\to\infty$ under $\alpha>0$ and $\beta>0$.
Consequently, the cumulative hazard satisfies

\[
H(t) \sim (\gamma/\beta)t\quad \text{as } t\to\infty,
\] 
\noindent implying exponential tail decay of the
corresponding survival distribution. It follows that $M_T(s)$ exists for all $s<\gamma/\beta$ and
diverges for $s\ge \gamma/\beta$.

For $\eta>0$, solutions of the population dynamics–based hazard equation converge to the equilibrium $K$ as $t \to \infty$.
Consequently, $H(t) \approx Kt$ as $t \to \infty$. This implies that
$S(t) \approx e^{-Kt}$ for large $t$, and that the moment generating
function $M_T(s)$ exists for all $s < K$ and diverges for $s > K$.

For the sinusoidal hazard model
with parameters
chosen such that $h(t)>0$ for all $t\ge 0$, the cumulative hazard can be
written as 
\[
H(t)=ct+O(1)\quad \text{as } t\to\infty.
\]

\noindent The oscillatory terms in $h(t)$ are bounded and therefore
contribute only bounded fluctuations to $H(t)$, without affecting its
linear growth rate. Consequently, the $S(t)$ decays
exponentially at rate $c$. This tail behavior determines the existence
of MGF. When $s<c$, the exponential decay of the survival distribution dominates the growth of the factor $e^{sT}$,
and the integral defining $M_T(s)$ converges. When $s\ge c$, the growth
of $e^{sT}$ overwhelms the tail decay, causing the integral to diverge.
As a result, $M_T(s)$ exists for all
$s<c$ and diverges for $s\ge c$.

For the exponential hazard model with an interaction term, when $\beta=0, \alpha>0$, $H(t)$ admits a
closed-form expression and grows exponentially in $t$, provided that
$h_0+v_0/\sqrt{\alpha}>0$. In particular,
\[
H(t)\sim C\,e^{\sqrt{\alpha}\,t}\quad \text{as } t\to\infty,
\]
for some constant $C>0$. Consequently, the rapid decay of the survival function implies that $M_T(s)$ exists and is finite for all $s\in\mathbb{R}$. In the boundary case $h_0+v_0/\sqrt{\alpha}=0$, the leading exponential term in $H(t)$ vanishes and $H(t)$ remains bounded. Consequently,
$S(\infty)>0$, implying a positive probability that the event does not
occur. In this case, $T=\infty$ with positive probability, and
$M_T(s)$ is infinite for all $s>0$.

%%%%%%%%%%%%%%%%%%%%%%%%%%%%%%%%%%%%%%%%%%
%%%%%%%%%%%%%%%%%%%%%%%%%%%%%%%%%%%%%%%%%%

\subsection{Maximum Likelihood Estimation}

Due to the parametric structure of the proposed hazard models, the log-likelihood can be expressed in terms of the hazard and cumulative hazard functions as
\begin{equation*}
	l(\bm{\theta}, \bm{Y}_0) = \sum_{i=1}^{n} \delta_i \log h(t_i \mid \bm{\theta}, \bm{Y}_0) - \sum_{i=1}^{n} H(t_i \mid \bm{\theta}, \bm{Y}_0),
\end{equation*}

Here, $\bm{\theta}$ denotes the vector of model parameters, specified to ensure that $h(t) > 0$ for all $t \ge 0$, and $\bm{Y}_0$ represents the vector of initial conditions that determine the starting state of the hazard dynamics. In this framework, the initial conditions are treated as unknown parameters and estimated jointly with $\bm{\theta}$, since they directly influence the trajectory of the hazard function over time. Although some information about these initial values may be inferred from early observations, specifying them a priori (e.g., fixing $h_0$) is generally not straightforward and may introduce bias or misspecification. This motivates incorporating $\bm{Y}_0$ into the parameter vector and estimating it alongside $\bm{\theta}$.

With this formulation, the log-likelihood can be evaluated for any given parameter values, allowing for maximum likelihood estimation(MLE) of both $\bm{\theta}$ and $\bm{Y}_0$ using general-purpose numerical optimization routines, such as \texttt{optim} or \texttt{nlminb} in \textsf{R}.

While MLE provides a natural starting point for inference, it may be sensitive to data limitations commonly encountered in survival analysis, such as censoring, limited follow-up, or small sample sizes. To address these challenges, a Bayesian approach offers a complementary framework by incorporating prior information and yielding full posterior distributions for the unknown parameters. However, a key component of the Bayesian approach is the specification of prior distributions for the model parameters, which may vary across different hazard structures. While prior selection can be nontrivial, the interpretability of the parameters often provides guidance for constructing informative or weakly informative priors. Given the resulting posterior distribution, inference requires efficient sampling methods. In this setting, posterior computation is tractable because the likelihood can be evaluated, either analytically or numerically, for any parameter configuration. This makes the proposed models well suited for implementation using general-purpose sampling algorithms. Accordingly, posterior inference can be carried out using standard MCMC methods, including Metropolis-within-Gibbs (e.g., \texttt{BUGS}, \texttt{spBayes}), Hamiltonian Monte Carlo (e.g., \texttt{Stan}), and adaptive samplers such as \texttt{twalk} and \texttt{MCMCpack}. Further details of the Monte Carlo implementation are provided in Appendix~\ref{sec:mc-algorithm}.

%%%%%%%%%%%%%%%%%%%%%%%%%%%%%%%%%%%%%%%%%%
%%%%%%%%%%%%%%%%%%%%%%%%%%%%%%%%%%%%%%%%%%

\section{Simulation Study}\label{sec:simulation}

\mbh{In this section, we evaluate the finite-sample performance of the proposed higher-order hazard models described in Section~\ref{sec:non-linear} and assess their ability to capture complex temporal risk patterns. We first examine parameter estimation accuracy under different sample sizes, and then conduct comparative studies against standard parametric models and ODE-based benchmarks related to the formulations of Christen and Rubio \cite{christen2024dynamic, christen2024harmonic}. In particular, the underdamped oscillatory model serves as a linear second-order ODE-based benchmark corresponding to their harmonic oscillator formulation.}

%In this section, we first conduct simulation studies to evaluate the finite-sample performance of the proposed hazard models described in Section~\ref{sec:non-linear}. We then perform a comparison study to assess the relative performance of the damped oscillatory hazard model and the sinusoidal hazard model.

\subsection{Parameter Estimation Performance}

 We generate event times from the proposed models under fixed parameter settings.

\begin{enumerate}
  \item Damped oscillatory hazard model,
  \begin{enumerate}
  \item Underdamped: $\alpha = 0.5, \ \beta = 1, \ \gamma =0.2, \ h_0 = 0.1, \ v_0 = 0.3$
  \item Critically damped: $ \alpha = 3, \ \gamma = 0.2, \ h_0 = 0.1, \ v_0 = 0.3$
  \item Overdamped: $\alpha = 3, \ \beta = 1, \ \gamma = 0.2, \ h_0 = 0.1, \ v_0 = 0.3$
  \end{enumerate}  
  \item Population dynamics–based hazard model: $ r = 0.8, \ \zeta = 0.5, \ K = 1, \ h_0 = 0.1, \ v_0 = 0.2$,
  \item Sinusoidal hazard model: $\omega = 0.2 \pi, \ c = 0.6, \ h_0 = 0.1, \ v_0= 0.2$
  \item Exponential hazard model with interaction effects : $ \alpha = 0.1, \ h_0 = 0.4, \ v_0 = 0.1 $
\end{enumerate}

To assess performance, 20,000 Monte Carlo replications were generated for each model using sample sizes of 500, 1000, 2000, and 5000. Independent right censoring was imposed by generating censoring times from a uniform distribution, with observed times defined as $t=\min(T,C)$ and event indicators $\delta=\mathbb{I}(T \le C)$, resulting in censoring rates of approximately $20\%$ to $30\%$.

For each simulated dataset, model parameters were estimated via  Bayesian framework using the t-walk sampler implemented in R. Weakly informative prior distributions were adopted for all model parameters; in particular, positive-valued parameters were assigned Gamma(2, 2) priors to provide mild regularization while retaining flexibility. To obtain posterior samples, 110,000 Markov chain Monte Carlo (MCMC) iterations were run. The first 10,000 iterations were discarded as burn-in, and a thinning interval of 5 was applied by retaining every fifth draw to reduce autocorrelation, yielding 20,000 samples for posterior inference. Model performance was summarized across Monte Carlo replications using posterior means and root mean squared errors (RMSE) to examine estimation accuracy under different sample sizes.

Table~\ref{tab:udamp_ps} reports the posterior means and RMSEs for each parameter of the underdamped hazard model. As expected, RMSEs decrease with increasing sample size, indicating improved estimation accuracy and consistency. Figure~\ref{fig:hist-u} displays the posterior distributions of the model parameters for the underdamped hazard model when the sample size is
$n=2000$. Compared to the corresponding priors (shown in red), the posterior distributions are more concentrated, reflecting increased information from the data at this sample size. Similar summaries for the remaining models proposed in section ~\ref{sec:non-linear} are provided in the Appendix ~\ref{add:sim.results}. \mbh{These results demonstrate that the proposed higher-order models can be reliably estimated from censored survival data, with estimation accuracy improving as the sample size increases.}

\begin{table}[H]
\centering
\caption{Posterior Results for the underdamped hazard model. True parameter values are shown in parentheses.}
\label{tab:udamp_ps}
\begin{tabular}{ccccccc}
\toprule
 &  & $\alpha\,(0.5)$ & $\beta\,(1)$ & $\gamma\,(0.2)$ & $h_0\,(0.1)$ & $v_0\,(0.3)$ \\
\midrule
$n=200$  & Mean  & 0.5821 & 1.0219 & 0.2003 & 0.0891 & 0.4339 \\
         & RMSE  & 0.2783 & 0.1860 & 0.0553 & 0.0473 & 0.1870 \\
\midrule
$n=500$  & Mean  & 0.7425 & 1.2398 & 0.2700 & 0.0728 & 0.3646 \\
         & RMSE  & 0.3813 & 0.3710 & 0.1080 & 0.0399 & 0.1251 \\
\midrule
$n=1000$ & Mean  & 0.4197 & 1.0985 & 0.2286 & 0.0975 & 0.2597 \\
         & RMSE  & 0.1791 & 0.1511 & 0.0452 & 0.0225 & 0.0714 \\
\midrule
$n=2000$ & Mean  & 0.3778 & 1.0697 & 0.2315 & 0.1010 & 0.2780 \\
         & RMSE  & 0.1540 & 0.0930 & 0.0372 & 0.0166 & 0.0456 \\
\midrule
$n=5000$ & Mean  & 0.5335 & 0.9635 & 0.1885 & 0.0914 & 0.3156 \\
         & RMSE  & 0.0778 & 0.0635 & 0.0186 & 0.0137 & 0.0326 \\
\bottomrule
\end{tabular}
\end{table}

%%%%%%%%%%%%%%%%%%%%%%%%%%%%%%%%%%%%%%%%%%

\begin{figure}[h!]
    \centering
    \includegraphics[width=\linewidth]{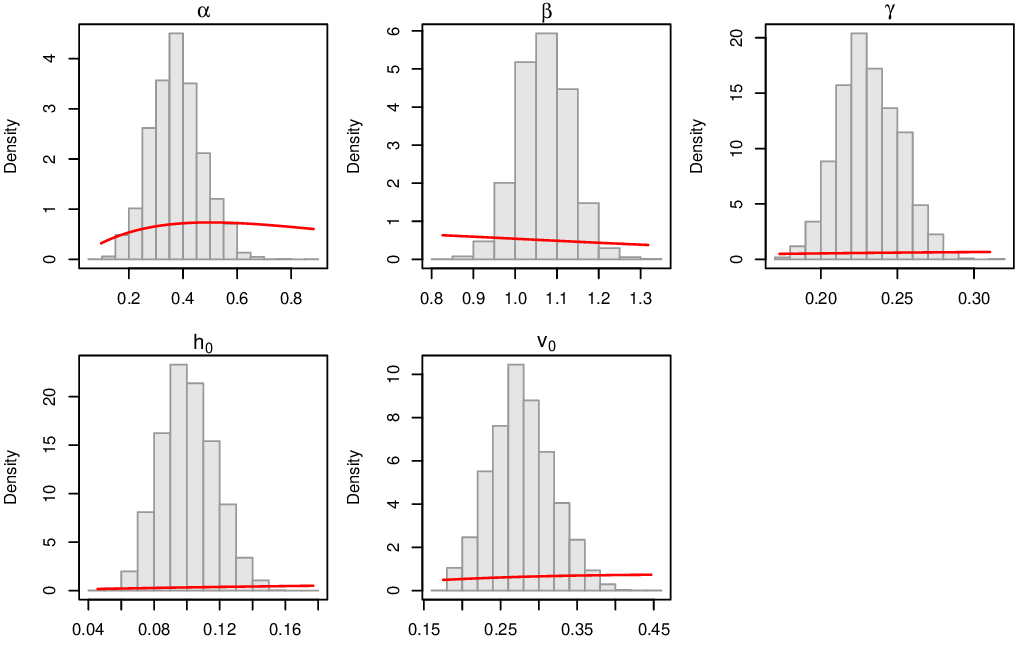}
    \caption{Posterior distributions (n=2000) for the underdamped hazard model parameters. The relevant part
of the prior density is shown in red.}
    \label{fig:hist-u}
\end{figure}

\subsection{Comparison of Hazard Models with the Sinusoidal Hazard Model}
\label{subsec:comparison1}

\cb{In this subsection, we compare the sinusoidal hazard model with several competing hazard models, including standard parametric models (Weibull and lognormal) and the ODE-based oscillatory hazard model presented in Christen and Rubio \cite{christen2024harmonic}, to demonstrate that the sinusoidal hazard uniquely captures sustained, purely oscillatory risk dynamics without damping, a feature that standard parametric and damped oscillatory models cannot represent.}

We generate event times from the sinusoidal hazard model with $\omega = 2\pi/t_c$, $c = 0.15$, $h_0 = 0.15$, and $v_0 = \omega c h_0 \approx 0.0943$. We set $t_c = 4.5$, so that the hazard completes one full periodic cycle every 4.5 years. To incorporate right censoring, censoring times are generated from a $\text{uniform}(0, c_{\max})$ distribution, where $c_{\max}$ is calibrated to achieve approximately $30\%$ censoring.
%For comparison, we fit several competing parametric hazard models, including the lognormal, Weibull, and underdamped oscillatory models, across different sample sizes. 
%\mbh{For comparison, we fit several competing hazard models, including standard parametric models (Weibull and lognormal) as well as ODE-based hazard models related to the formulations of Christen and Rubio \cite{christen2024harmonic}. In particular, the underdamped oscillatory model corresponds to their linear second-order oscillator framework and serves as a representative linear ODE-based benchmark.}
We then fit the competing models to the simulated data using a Bayesian estimation approach for both the sinusoidal and oscillatory hazard models. The initial values for $h_0$ and $c_0$ are set to the empirical average hazard rate, and $v_0$ is initialized at $0.001$. The frequency parameter $\omega$ is fixed at its true value. In practice, $\omega$ may be informed by prior knowledge or by examining previously observed temporal patterns in the data. For the underdamped oscillatory model, we set $\alpha = 0.001$ and choose $\beta$ accordingly to satisfy the underdamping condition. The performance of these models is evaluated using the Bayesian information criterion (BIC):
\[
BIC = k \log(n) - 2l
\]
where $n$ is the sample size, $k$ is the number of parameters, and $\ell$ is the maximized log-likelihood. The results are summarized in Table~\ref{tab:bic}, and as expected, when the underlying hazard exhibits stable cyclic behavior with constant amplitude, the BIC favors the sinusoidal hazard model, indicating a more appropriate fit. This occurs because the underdamped oscillatory hazard model represents periodic behavior with a decaying amplitude over time, whereas the sinusoidal hazard model captures periodic patterns with constant amplitude. Figure~\ref{fig:ht.comparison} illustrates the estimated hazard functions from the competing models along with the true hazard for $n=2000$, showing that the sinusoidal model closely follows the underlying periodic pattern. This highlights the importance of selecting a model whose dynamical structure aligns with the underlying risk mechanism, a key advantage of the proposed framework. \mbh{In this sense, the comparison with the linear oscillatory benchmark highlights the additional flexibility of the proposed framework in representing periodic hazard behavior beyond decaying oscillatory dynamics.}

\begin{table}[ht]
\centering
\caption{BIC comparison of hazard models with Sinusoidal hazard model under varying sample sizes}
\label{tab:bic}
\begin{tabular}{lccc}
\hline
Model & $n=500$ & $n=1000$ & $n=2000$ \\
\hline
Sinusoidal hazard model & 2101.21 & 4204.88 & 8450.89\\
Underdamped oscillatory model$^{*}$ & 2112.19 & 4220.57 & 8461.73 \\
Weibull model & 2109.18 & 4218.99 & 8535.79\\
Lognormal model & 2160.43 & 4296.33 & 8662.27 \\
\hline
\end{tabular}
\vspace{2pt}

{\footnotesize $^{*}$Comparison with ODE-based oscillatory hazard model in Christen and Rubio \cite{christen2024harmonic}.}
\end{table}

\begin{figure}[h!]
    \centering
    \includegraphics[scale=0.7]{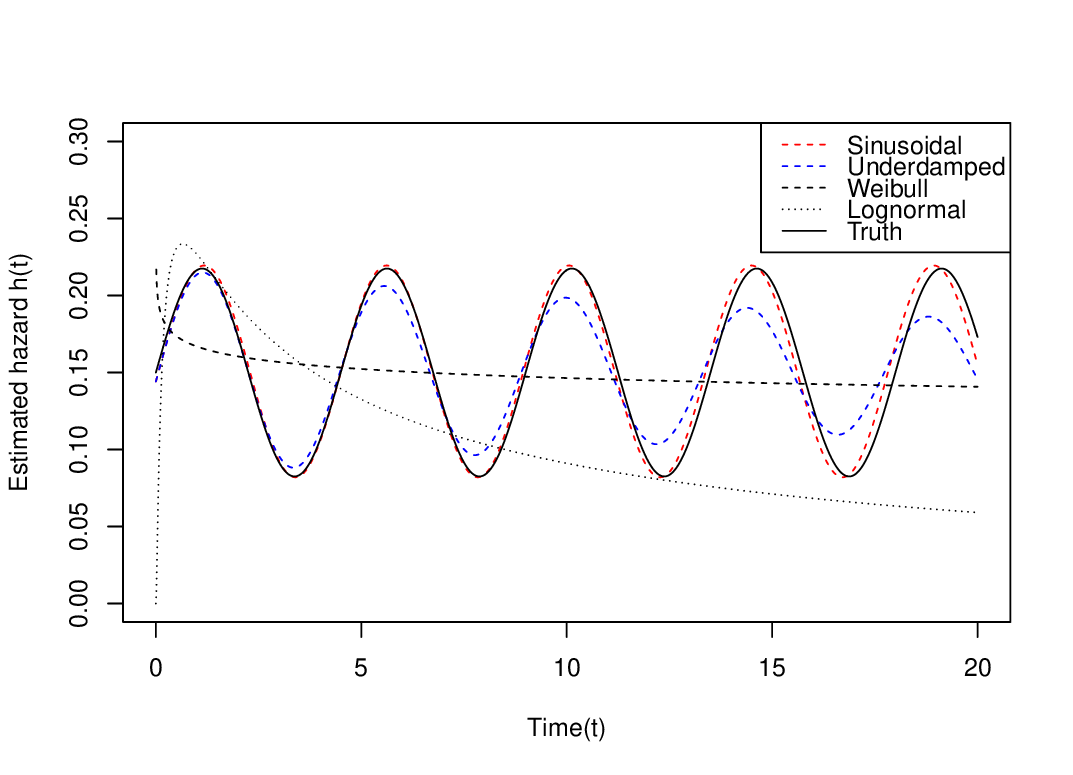}
    \caption{Comparison of competing hazard functions with Sinusoidal hazard model ($n=2000$).}
    \label{fig:ht.comparison}
\end{figure}

\subsection{Comparison of Hazard Models with the Population Dynamics–Based Hazard Model}

In this subsection, we compare the proposed population dynamics-based hazard model (second-order logistic) with the first-order logistic and underdamped oscillatory hazard models through simulation. The study focuses on a scenario with transient non-monotone hazard dynamics, where the second-order formulation captures behavior not adequately represented by monotone logistic or purely oscillatory models.

We generate event times from the population dynamics--based hazard model with $r = 0.3$, $\zeta = 1.4$, $K = 0.03$, $h_0 = 0.06$, and $v_0 = 0.05$. To incorporate right censoring, censoring times are generated from a uniform  distribution, resulting in approximately $30\%$ right censoring. We used a Bayesian approach to estimate the hazard models. As described earlier (Subsection \ref{subsec:comparison1}), initial parameter values were specified for the underdamped oscillatory hazard model. For the population dynamics–based model and the logistic hazard model, initial parameter values were chosen using simple summaries of the observed event times to aid numerical optimization. The carrying capacity parameter was initialized using the inverse of the first quartile of the event times, while the growth parameter was initialized using the inverse squared median event time. The initial hazard level was set using the empirical event rate, and a small positive value (0.001) was used for both the initial velocity $v_0$ and the damping parameter $\zeta$. Table~\ref{tab:bic2} presents the BIC values for the competing hazard models across different sample sizes. The population dynamics–based hazard model consistently attains the lowest BIC, indicating superior fit compared \mbh{compared to both standard parametric models and competing ODE-based formulations, including the linear oscillatory benchmark \cite{christen2024harmonic}. This demonstrates the advantage of nonlinear higher-order dynamics in capturing transient, non-monotone hazard behavior that cannot be adequately represented by monotone or purely oscillatory models.}

\begin{table}[ht]
\centering
\caption{BIC comparison of hazard models with Population dynamics-based model under varying sample sizes}
\label{tab:bic2}
\begin{tabular}{lccc}
\hline
Model & $n=500$ & $n=1000$ & $n=2000$ \\
\hline
Population dynamics-based model & 3018.38 & 6093.31 & 12073.78\\
Logistic growth model$^{^{\dagger}}$ & 3123.82 & 6121.31 &  12135.57\\
Underdamped oscillatory model$^{*}$ & 3022.14 & 6102.65 & 12091.49\\
Weibull model & 3030.36 & 6139.74 & 12183.75\\
Lognormal model & 3070.07 & 6222.96 & 12334.44\\
\hline
\end{tabular}

{\footnotesize $^{^{\dagger}}$Comparison with Logistic growth hazard model in Christen and Rubio \cite{christen2024dynamic}.\\$^{*}$Comparison with ODE-based oscillatory hazard model in Christen and Rubio \cite{christen2024harmonic}.}
\end{table}

\begin{figure}[h!]
    \centering
    \includegraphics[scale = 0.7]{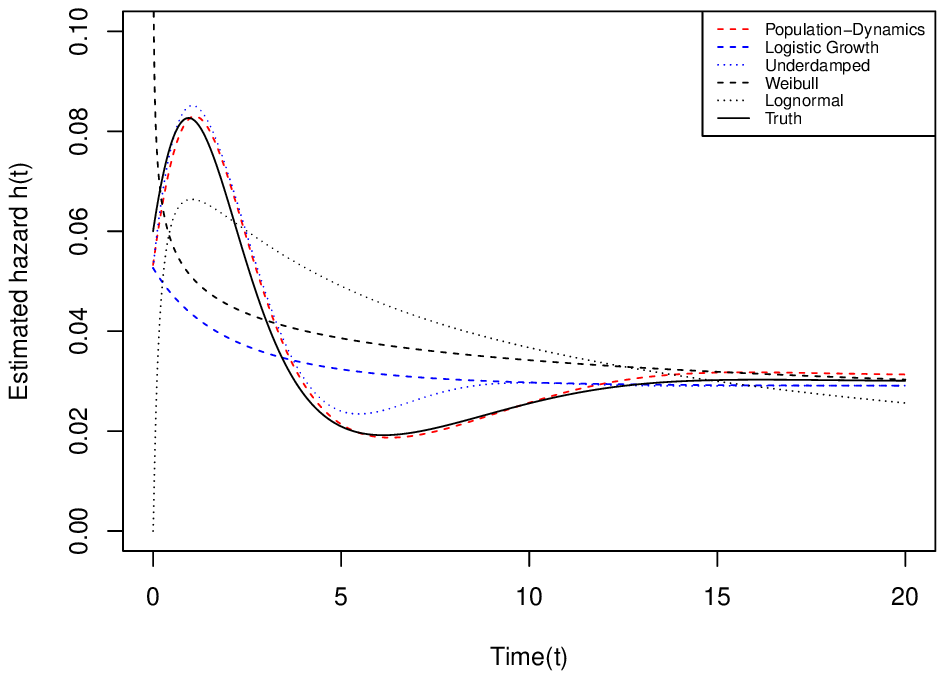}
    \caption{Comparison of competing hazard functions with Population dynamics-based hazard model ($n=2000$).}
    \label{fig:ht.comparison2}
\end{figure}

%%%%%%%%%%%%%%%%%%%%%%%%%%%%%%%%%%%%%%%%%%
%%%%%%%%%%%%%%%%%%%%%%%%%%%%%%%%%%%%%%%%%%
\section{Real Data Application}\label{sec:real_data}

In this section, we analyze a dataset from a clinical trial in gastric oncology that contains right-censored survival data. The gastric cancer data were collected by the Gastrointestinal Tumor Study Group in 1982 and are freely accessible through the R package \texttt{coxphw}. The dataset includes information on 90 patients diagnosed with locally advanced gastric cancer. First, we inspected the observed event times to assess the pattern of failures. The events appear to be concentrated in the early to intermediate follow-up period, with fewer occurring later, suggesting that the risk of failure may vary over time rather than remaining constant. 

Standard parametric models such as the Weibull and lognormal distributions impose restrictive hazard shapes, with the Weibull allowing only monotone hazards and the lognormal producing a single peak. These forms may not adequately capture more complex time-varying behavior. Therefore, more flexible hazard formulations, such as sinusoidal and oscillatory models, may be appropriate for describing the underlying risk pattern. Therefore, we employ underdamped oscillatory, sinusoidal, and population dynamics hazard models and compare them with logistic growth, Weibull, and lognormal hazard models.

To choose the initial parameters, we employ simple data-driven values that are consistent with the local interpretation of the model. Specifically, the initial parameter values are chosen to correspond to the survival function values $S(\Delta t)$, $S(2\Delta t)$, and $S(3\Delta t)$ for a small time step $\Delta t$. In this analysis, $\Delta t$ is selected as two weeks. The corresponding survival probabilities are computed directly from the observed data using the proportion of individuals who have not experienced the event by the given time. Since the sample size is $n=90$ and the cumulative number of observed failures by days $14$, $28$, and $42$ is $1$, $2$, and $3$, respectively, the empirical survival probabilities are
\[
S_0 = 1, \qquad 
S_1 = \frac{89}{90} = 0.9889, \qquad 
S_2 = \frac{88}{90} = 0.9778, \qquad 
S_3 = \frac{87}{90} = 0.9667.
\]

Using forward finite-difference approximations of the survival function, we estimate the initial hazard $h_0 = h(0)$ and its first derivative $v_0 = h'(0)$. Let $S_0 = S(0)$, $S_1 = S(\Delta t)$, and $S_2 = S(2\Delta t)$. Then
\[
S'(0) \approx \frac{S_1 - S_0}{\Delta t}, \qquad 
S''(0) \approx \frac{S_2 - 2S_1 + S_0}{\Delta t^2}.
\]

Using the relationship $h(t) = -\frac{S'(t)}{S(t)}$, the initial hazard and its derivative are approximated by
\[
h_0 \approx -\frac{S_1 - S_0}{\Delta t\,S_0}, \qquad
v_0 = h'(0) \approx 
\left(\frac{S_1 - S_0}{\Delta t\,S_0}\right)^2 
- \frac{S_2 - 2S_1 + S_0}{\Delta t^2\,S_0}.
\]

These approximations provide reasonable starting values that reflect the early-time behavior of the hazard function implied by the observed data.

For the sinusoidal model, we use the relationship between the model parameters and the derivatives of the hazard function at $t=0$. In particular, the parameter $\omega$ is obtained from the model definition
\[
\omega^2 = - \frac{h''(0)}{h_0 - c}.
\]

For the oscillatory hazard model, the parameter $\omega$ is chosen based on the spread of the observed event times using the interquartile range. A small positive value is selected for $\alpha = 0.001$, and $\beta$ is then determined from $\omega$ and $\alpha$. The parameter $\gamma$ controls the overall scale of the hazard function and is chosen so that the model reflects the empirical hazard rate estimated from the first two weeks of follow-up.

For the logistic growth and population dynamics hazard models, the carrying capacity parameter $K$ is initialized from the tail of the observed event times. Specifically, we consider observations beyond the third quartile and estimate the hazard rate in this period as the number of failures divided by the total person-time at risk. The growth parameter $r$ is initialized using the reciprocal of the squared mean event time. Given the previously estimated values of $h_0$, $h'(0)$, and $h''(0)$, the damping parameter $\zeta$ is then obtained from the model relationship so that the early behavior of the hazard function matches the empirical derivatives estimated from the data.

The initial parameter values were first used to obtain MLEs via \texttt{nlminb}, which were subsequently used as starting values for the Bayesian analysis implemented with \texttt{twalk}. Model performance was evaluated using the Akaike Information Criterion (AIC), the corrected Akaike Information Criterion (AICc), and the Bayesian Information Criterion (BIC). Because the sample size was moderate ($n = 90$), AICc was included to account for potential small-sample bias \citep{Burnham2002}. As shown in Table \ref{tab:ra.model.comp}, the sinusoidal hazard model has the smallest AIC and AICc values, indicating the best fit among the models considered, while the Weibull model has the lowest BIC due to its simpler structure. The hazard functions implied by the fitted models are shown in Figure~\ref{fig:real.comparison}.

\begin{table}[H]
\centering
\caption{Comparison of hazard models for the gastric cancer data based on AIC, AICc, and BIC.}
\begin{tabular}{lccc}
\hline
Model & AIC & AICc & BIC \\
\hline
Sinusoidal model              & 251.29 & 251.76 & 260.29 \\
Underdamped oscillatory model & 255.22 & 255.93 & 267.72 \\
Logistic growth model         & 260.62 & 260.90 & 268.12 \\
Population dynamics-based model & 255.92 & 256.63 & 268.41 \\
Weibull model                 & 254.87 & 255.01 & 259.87 \\
Lognormal model               & 262.06 & 262.19 & 267.05 \\
\hline
\end{tabular}

{\footnotesize $^{^{\dagger}}$Comparison with Logistic growth hazard model in Christen and Rubio \cite{christen2024dynamic}.\\$^{*}$Comparison with ODE-based oscillatory hazard model in Christen and Rubio \cite{christen2024harmonic}.}
\label{tab:ra.model.comp}
\end{table}

\begin{figure}[H]
    \centering
    \includegraphics[scale=0.65]{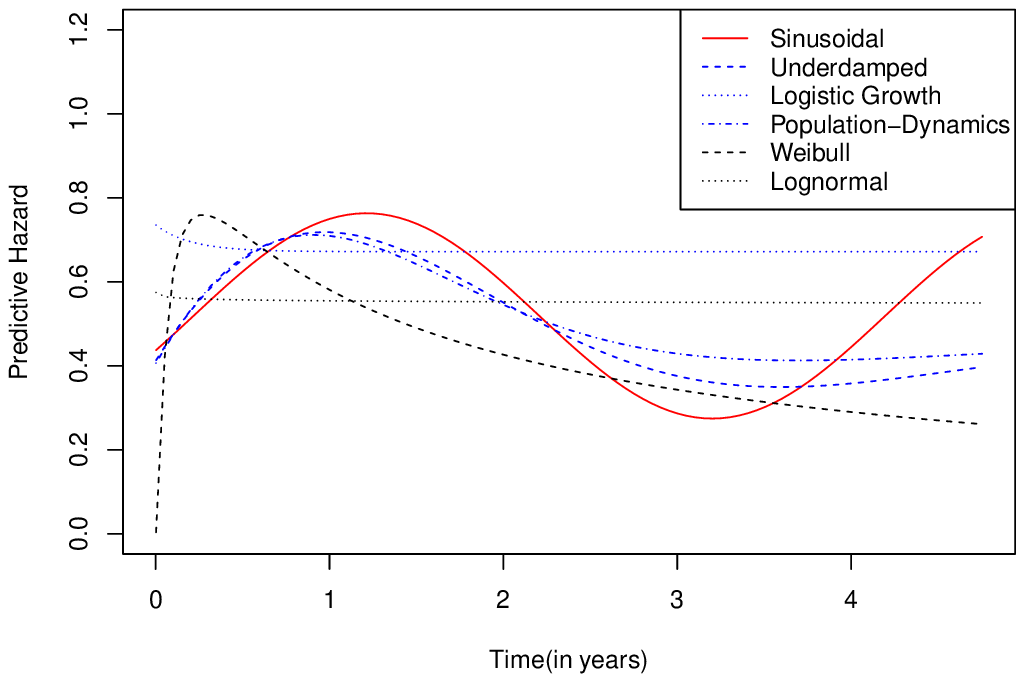}
    \caption{Estimated hazard functions for the fitted models using the gastric cancer dataset.}
    \label{fig:real.comparison}
\end{figure}

\mbh{Overall, these results demonstrate that the proposed higher-order hazard models provide a flexible and effective framework for modeling complex time-varying risk patterns in real-world data, offering improvements over both classical parametric models and linear oscillatory ODE-based benchmarks \cite{christen2024harmonic}, particularly in capturing non-monotonic and nonlinear hazard dynamics.}

%Overall, these results demonstrate that the proposed higher-order hazard models provide a flexible and effective framework for modeling complex time-varying risk patterns in real-world data, offering improvements over both classical parametric models and simpler ODE-based formulations.

\section{Conclusion}\label{sec: Dis_conclusion}

This work advances dynamical survival analysis by introducing a class of hazard models governed by higher-order ODEs. By allowing the hazard to evolve according to both its current level and its temporal derivatives, the proposed framework substantially broadens the range of risk dynamics that can be represented within a coherent survival modeling paradigm. This is particularly relevant in medical and biomedical applications, where risk evolves over time in response to treatment, disease progression, or physiological adaptation.  \mbh{In contrast to conventional first-order or monotone hazard models, higher-order formulations accommodate inertia, delayed responses, and oscillatory behavior, yielding hazard trajectories and survival distributions that are otherwise inaccessible.}

Through a collection of nonlinear and oscillatory examples, we demonstrate how interpretable dynamical mechanisms translate directly into flexible hazard behavior and complex time-to-event patterns. The accompanying numerical framework enables practical implementation across model classes, supporting solution of higher-order systems, evaluation of cumulative hazards, likelihood-based inference under censoring, and simulation of event times. Together, these components establish higher-order ODE-based hazards as a viable and expressive alternative to static or memoryless survival models, particularly in applications characterized by feedback-driven or temporally adaptive risk.

Despite its flexibility, the proposed framework presents several challenges. Higher-order hazard models introduce additional parameters and latent dynamics, which may complicate identifiability and increase computational cost, particularly in data-limited settings. Model specification requires careful consideration of parameter choices, dynamical structure, and time-scale normalization to ensure the existence of a positive hazard function and meaningful temporal interpretation. In particular, the inclusion of an explicit scale parameter is essential for preserving hazard shape across different time scales~\cite{christen2024dynamic}, consistent with standard parametric survival modeling practice. Inappropriate modeling choices may lead to instability or overfitting. Moreover, the present work focuses primarily on deterministic hazard dynamics, and stochastic perturbations or measurement noise are not explicitly modeled.

These limitations point to promising directions for future work. \mbh{Extensions to stochastic differential equation formulations would allow the incorporation of random fluctuations in hazard dynamics, while the inclusion of covariates would enhance applicability in regression settings. Additional work is also needed to further investigate identifiability, model selection, and computational efficiency in high-dimensional settings.}

Overall, higher-order dynamical hazard models offer a principled and flexible approach for modeling complex temporal risk patterns, \mbh{offering clear advantages over both classical parametric models and existing ODE-based formulations in settings where risk evolves through feedback, memory, and adaptation.}

%%%%%%%%%%%%%%%%%%%%%%%%%%%%%%%%%%%%%%%%%%
\subsection*{Code Availability}

The code used in this study is available upon request.

%%%%%%%%%%%%%%%%%%%%%%%%%%%%%%%%%%%%%%%%%%
\subsection*{Declaration of Interests}
The author has no conflicts of interest to report.

%%%%%%%%%%%%%%%%%%%%%%%%%%%%%%%%%%%%%%%%%%
\subsection*{Contributions}
Authors DL, MBH, and FM contributed to the formulation and analysis of the models, graphical presentations, data curation, and the simulation study; FM conceptualized and directed the research. All authors contributed to the literature search, model evaluations, and the writing of the paper.

%%%%%%%%%%%%%%%%%%%%%%%%%%%%%%%%%%%%%%%%%%
\subsection*{Ethics Approval}

There is no ethical approval needed due to the use of simulated and publicly available data.

%%%%%%%%%%%%%%%%%%%%%%%%%%%%%%%%%%%%%%%%%%
\subsection*{Funding Statement}

The authors do not have funding to report.

%%%%%%%%%%%%%%%%%%%%%%%%%%%%%%%%%%%%%%%%%%
\setlength{\bibsep}{0pt}
\bibliographystyle{unsrt}
\bibliography{bibfile}

\newpage
%%%%%%%%%%%%%%%%%%%%%%%%%%%%%%%%%%%%%%%%%%
%%%%%%%%%%%%%%%%%%%%%%%%%%%%%%%%%%%%%%%%%%
\appendix
%%%%%%%%%%%%%%%%%%%%%%%%%%%%%%%%%%%%%%%%%%
%%%%%%%%%%%%%%%%%%%%%%%%%%%%%%%%%%%%%%%%%%
\section{Monotonicity and Nonmonotonicity in ODEs}

For first-order ODEs, monotonicity of solutions is closely tied to the sign of the derivative along solution trajectories. Consider an autonomous first-order ODE
\[
h'(t) = \psi(h(t)), \quad h(t_0) = h_0.
\]
If \( \psi(h(t)) \) maintains a constant sign along the trajectory of the solution, then \( h(t) \) is monotonic on its interval of existence. For example, when \( \psi(h) = k h \) with \( k > 0 \), the solution \( h(t) = h_0 e^{k(t-t_0)} \) is strictly increasing. Conversely, if \( \psi(h(t)) \) changes sign along the trajectory, the solution may exhibit nonmonotonic behavior, such as local maxima or minima, depending on the structure of \( \psi \).

First-order ODEs have been widely used to model hazard dynamics in survival analysis, providing a unified perspective on classical models such as the proportional hazards, linear transformation, and accelerated failure time models \cite{cox2007parametric, liu2012survival}. Within this framework, the evolution of the hazard is governed solely by its current level, which imposes strong structural constraints on admissible temporal behavior. In particular, first-order autonomous systems naturally favor monotonic or asymptotically monotonic hazard trajectories.

In contrast, monotonicity becomes substantially more restrictive in higher-order ODEs. For instance, consider the second-order linear ODE
\[
h''(t) = -k h(t), \quad h(t_0) = h_0, \quad h'(t_0) = v_0,
\]
with \( k > 0 \). Its solutions take the form
\[
h(t) = A \cos(\sqrt{k}t) + B \sin(\sqrt{k}t),
\]
which are inherently oscillatory and therefore nonmonotonic. More generally, higher-order ODEs introduce additional state variables, such as the hazard derivative, that act as latent dynamical components. Even when the system is autonomous, these hidden states permit feedback, inertia, and delayed responses, making nonmonotonic behavior generic rather than exceptional.

Nonlinear higher-order systems further expand the range of possible dynamics. For example, the nonlinear second-order equation
\[
h''(t) + h(t)^3 = 0
\]
admits bounded oscillatory solutions whose amplitude and frequency depend on initial conditions. Such dynamics can generate transient risk increases, cyclic hazard patterns, and delayed effects that cannot be easily captured within first-order formulations.

From the perspective of hazard modeling, nonmonotonic hazard functions may arise either through explicit time dependence or through higher-order dynamics that implicitly encode memory via additional state variables. Second-order and higher-order ODEs thus provide a natural mechanism for generating nonmonotone hazards without requiring nonautonomous forcing. While many commonly used hazard models can be expressed as solutions to first-order equations under specific boundary conditions such as \( \lim_{t \to 0} h(t) = 0 \) or \( \infty \), higher-order formulations substantially extend this class by enabling richer temporal behavior.

%%%%%%%%%%%%%%%%%%%%%%%%%%%%%%%%%%%%%%%%%%
\section{Analytical Solutions of the Damped Oscillatory Hazard Model}\label{app:closeform_damped}

\subsubsection*{Case 1: Underdamped $(\Delta <0)$}

\noindent For the underdamped case $ \Delta = \alpha^2 - 4\beta < 0 $, the characteristic equation has complex conjugate roots,
\[
r_{1,2} = -\frac{\alpha}{2} \pm i\omega, \quad \omega=\tfrac12\sqrt{4\beta-\alpha^{2}}.
\]
\noindent To solve Eq.~\eqref{eq:osc.ht}, the general solution for $h(t)$ in the case $(\Delta <0)$ can be written as

\[
h(t)=e^{-\frac{\alpha t}{2}}\big(A\cos(\omega t)+B\sin(\omega t)\big)+h^*, 
\quad h^*=\frac{\gamma}{\beta},
\]

\noindent where $h^*$ is the equilibrium hazard. Applying the initial conditions $h_0 = h(0)$ and $ v_0 = h^\prime(0)$, the constants $A$ and $B$ are obtained as

\[
\begin{aligned}
A &= h_0 - \dfrac{\gamma}{\beta}, \quad 
B = \dfrac{1}{\omega} \left(v_0 + \dfrac{\alpha}{2} \big(h_0 -  \dfrac{\gamma}{\beta}\big)\right).
\end{aligned}
\]

\noindent The value $h_0$ represents the initial hazard level, and $v_0$ denotes the initial rate of change of the hazard. The coefficient $\alpha$ is the damping parameter, which determines the rate at which the oscilations decay over time. The term $\omega$ is the angular frequency that determines the speed of oscillation, and $\gamma/ \beta$ is the long-run equilibrium hazard level.

\subsubsection*{Case 2: Critically Damped (\( \Delta = 0 \))}

\noindent In the critically damped case \( \Delta = 0 \), the characteristic equation has a repeated root $r=-\alpha /2$. The solution of Eq.~\eqref{eq:osc.ht} therefore takes the form:

\[
h(t) = (A + Bt)e^{-\alpha t/2} + \frac{\gamma}{\beta}.
\]

\noindent The constants $A$ and $B$ follow from the intial conditions $h_0 = h(0)$ and $ v_0 = h^\prime(0)$:

\[
A = h_0 - \dfrac{\gamma}{\beta}, 
\quad
B = v_0 + \dfrac{\alpha}{2} \left(h_0 - \dfrac{\gamma}{\beta} \right)
\]

\noindent In this setting, $\alpha$ serves as the critical damping coefficient and determines the rate at 
which the hazard approaches equilibrium without oscillation. Under the critical damping condition, 
the restoring coefficient $\beta$ satisfies $\beta = \alpha^{2}/4$.

\subsubsection*{Case 3: Overdamped (\( \Delta > 0 \))}

\noindent For the overdamped case \( \Delta > 0 \), the characteristic equation has two real and distinct roots,
\[
r_{1,2} = -\frac{\alpha}{2} \pm \sqrt{\frac{\alpha^{2}}{4} - \beta}.
\]

\noindent The corresponding solution to Eq.~\eqref{eq:osc.ht} is
\[
h(t) = A e^{r_{1} t} + B e^{r_{2} t} + \frac{\gamma}{\beta},
\]

\noindent where the constants $A$ and $B$ are determined by the initial conditions 
$h_{0} = h(0)$ and $v_{0} = h'(0)$. In this case, the hazard returns to its equilibrium level $\gamma/\beta$ without oscillation, 
but at a slower rate than under critical damping. The two exponential terms reflect two 
different decay speeds associated with the distinct roots $r_{1}$ and $r_{2}$.

For each case described, the cumulative hazard function $H(t)$ is obtained by integrating the corresponding hazard function $h(t)$. In the underdamped case $(\Delta < 0)$, this yields:

\[
  H(t) = \int_0^t \left[ e^{-\alpha u/2} \big(A \cos(\omega u) + B \sin (\omega u)\big) + \dfrac{\gamma}{\beta}\right] \, du. 
\]
\noindent This expression simplifies to the following closed form:

\[
\begin{aligned}
 H(t) &= \dfrac{\gamma}{\beta}\,t
+\dfrac{A}{\beta}\left[
\dfrac{\alpha}{2}
+e^{-\alpha t/2}\left(
-\dfrac{\alpha}{2}\cos(\omega t)+\omega\sin(\omega t)
\right)\right]\\
&+\dfrac{B}{\beta}\left[
\omega
+e^{-\alpha t/2}\left(
-\dfrac{\alpha}{2}\sin(\omega t)-\omega\cos(\omega t)
\right)\right].
\end{aligned}
\]   

\noindent The density function is then obtained by combining the hazard and cumulative hazard functions as $f(t) = h(t)\exp\{-H(t)\}$. Then, using the inverse transform sampling approach, event times can be simulated for the underdamped case. A similar numerical procedure applies to the remaining cases, with the specific implementation depending on the forms of $h(t)$ and $H(t)$.

%%%%%%%%%%%%%%%%%%%%%%%%%%%%%%%%%%%%%%%%%%
\section{Proof of Theorem~\ref{thm:osc.longterm}}\label{app:osc-proof}

\begin{proof}
\begin{enumerate}

\item Eq.~\ref{eq:osc.ht} is a linear ODE with constant coefficients. By standard existence and uniqueness theory for linear systems, for any $h_0>0$ and $v_0\in\mathbb{R}$ there is a unique solution $h(t)$ defined for all $t\ge 0$.

\item Set $h^*=\gamma/\beta$ and define $y(t)=h(t)-h^*$. Substituting gives
\[
y''(t)+\alpha y'(t)+\beta y(t)=0.
\]
Its characteristic equation is $r^2+\alpha r+\beta=0$, with roots
\[
r_{1,2}=\frac{-\alpha\pm\sqrt{\alpha^2-4\beta}}{2}.
\]
If $\alpha>0$ and $\beta>0$, then $\Re(r_1),\Re(r_2)<0$, so $y(t)\to 0$ as $t\to\infty$, hence $h(t)\to h^*=\gamma/\beta$.

\item The qualitative form follows from the roots: if $\Delta=\alpha^2-4\beta<0$ the roots are complex with real part $-\alpha/2$, giving damped oscillations; if $\Delta=0$ there is a repeated root $-\alpha/2$, giving critical damping; if $\Delta>0$ there are two distinct negative real roots, giving an overdamped sum of decaying exponentials.
\end{enumerate}
\end{proof}

%%%%%%%%%%%%%%%%%%%%%%%%%%%%%%%%%%%%%%%%%%
\section{Linear Case of the Exponential Hazard Model}
\label{app:exponential}

When $\beta = 0$, the interaction term vanishes and the model reduces to the linear second-order equation
\begin{equation}
\label{eq:hazard-exp0}
 h''(t) = \alpha h(t).   
\end{equation}

For $\alpha > 0$, the solution of Eq.~\eqref{eq:hazard-exp0} is a linear combination of exponentially increasing and decreasing components, and the long-term behavior of $h(t)$ is dominated by the exponentially increasing component. When $\alpha < 0$, the solution is oscillatory. We therefore restrict attention to $\alpha > 0$, which isolates the exponential mechanism and provides a natural reference for comparison with the model in Eq.~\eqref{eq:expo_grow} when $\beta > 0$. The general solution of Eq.~\eqref{eq:hazard-exp0} in this case can be written as
\[
h(t) = A e^{\sqrt{\alpha}\, t} + B e^{-\sqrt{\alpha}\, t},
\]
where the constants $A$ and $B$ are determined by the initial conditions
$h(0)=h_0$ and $h'(0)=v_0$. Solving for $A$ and $B$ gives
\[
A = \frac{1}{2}\left(h_0 + \frac{v_0}{\sqrt{\alpha}}\right),
\qquad
B = \frac{1}{2}\left(h_0 - \frac{v_0}{\sqrt{\alpha}}\right).
\]
Thus,
\[
h(t)
= \frac{1}{2}\left(h_0 + \frac{v_0}{\sqrt{\alpha}}\right)e^{\sqrt{\alpha}\, t}
+ \frac{1}{2}\left(h_0 - \frac{v_0}{\sqrt{\alpha}}\right)e^{-\sqrt{\alpha}\, t}.
\]

\noindent For $\alpha>0$, the sign of $h(t)$ is determined by the coefficients in its exponential representation. Since $e^{\pm\sqrt{\alpha}t}>0$ for all $t\ge0$, $h(t)$ remains positive for all $t\ge0$ if and only if
\[
h_0+\frac{v_0}{\sqrt{\alpha}} \ge 0
\quad\text{and}\quad
h_0-\frac{v_0}{\sqrt{\alpha}} \ge 0,
\]
or equivalently,
\[
h_0 \ge \frac{|v_0|}{\sqrt{\alpha}}.
\]
If this condition is violated, one of the coefficients becomes negative and
the resulting solution eventually takes negative values, which is
incompatible with the definition of the hazard function. Moreover, when
$h_0>0$, $\alpha>0$, and $v_0<0$ satisfy $h_0=-v_0/\sqrt{\alpha}$, the
coefficient of the growing exponential term vanishes, and the hazard reduces
to
\[
h(t)=h_0 e^{-\sqrt{\alpha}t}.
\]
which is positive and strictly decreasing for all $t\ge0$ (see Figure~\ref{fig:plot-exp0-1}). As $t\to\infty$, $h(t)$ converges to zero and $H(t)$ approaches the finite limit $h_0/\sqrt{\alpha}$. Consequently, $S(t)$ converges to a positive constant, indicating an improper survival distribution with a nonzero long-term survival probability. With $h(t)$ available in closed form, the cumulative hazard function for Eq ~\eqref{eq:hazard-exp0} is given by
\begin{equation*}
    H(t) = \frac{1}{2\sqrt{\alpha}} \left(h_0 + \frac{v_0}{\sqrt{\alpha}}\right) \left( e^{\sqrt{\alpha} t} - 1 \right)
+
\frac{1}{2\sqrt{\alpha}} \left(h_0 - \frac{v_0}{\sqrt{\alpha}}\right) \left( 1 - e^{-\sqrt{\alpha} t} \right)
\end{equation*}

\noindent Using $h(t)$ and $H(t)$, the pdf of the event time can be derived and used for simulation.

%%%%%%%%%%%%%%%%%%%%%%%%%%%%%%%%%%%%%%%%%%
\begin{figure}[H]
    \centering
    \includegraphics[scale=0.7]{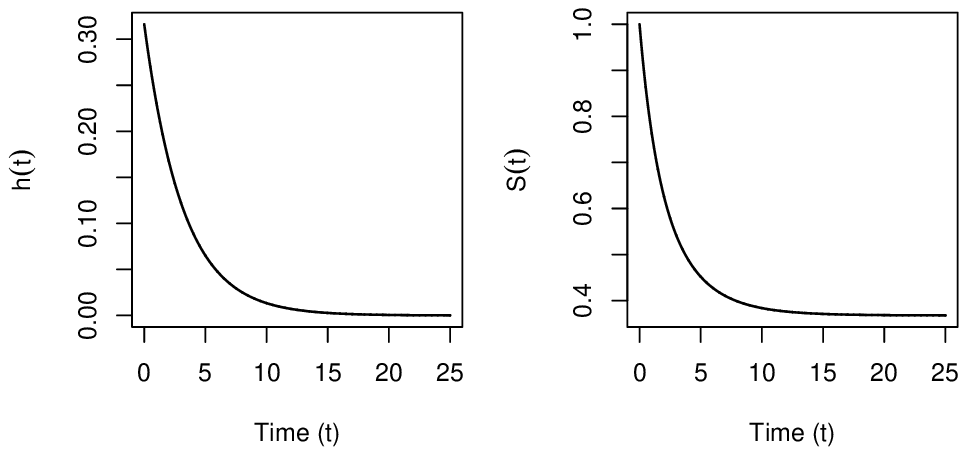}
    \caption{Hazard function $h(t)$ and survival function $S(t)$ for the exponential hazard model ($\alpha = 0.1,\ v_0 = -0.1$).}
    \label{fig:plot-exp0-1}
\end{figure}
%%%%%%%%%%%%%%%%%%%%%%%%%%%%%%%%%%%%%%%%%%

%%%%%%%%%%%%%%%%%%%%%%%%%%%%%%%%%%%%%%%%%%
\section{Monte Carlo Simulation for Nonlinear Hazard Models}
\label{sec:mc-algorithm}

\begin{algorithm}[H] 
\caption{Monte Carlo Simulation for Nonlinear Hazard Models}
\begin{algorithmic}[2]
\REQUIRE Parameter settings for each hazard model, sample sizes $n \in \{500, 1000, 2000, 5000\}$, number of replications $R = 20{,}000$
\ENSURE Posterior summaries and RMSEs for each model and sample size

\FOR{each hazard model (Damped Oscillatory, Population Dynamics, Sinusoidal, Exponential with interaction)}
    \STATE Set true parameter values according to model specification
    \FOR{each sample size $n$}
        \FOR{replication $r = 1$ to $R$}
            \STATE Simulate $n$ true event times $\{T_i\}_{i=1}^n$ from the model using numerical solver of second-order ODE
            \STATE Generate $n$ censoring times $\{C_i\}_{i=1}^n$ independently from a $\text{Uniform}(0, c_{\max})$ distribution, tuned to yield 20–30\% censoring
            \STATE Define observed times: $t_i = \min(T_i, C_i)$ and event indicators: $\delta_i = \mathbb{I}(T_i \le C_i)$
            \STATE Fit the model using Bayesian inference:
                \begin{itemize}
                    \item Use weakly informative priors (e.g., $\text{Gamma}(2,2)$ for positive parameters)
                    \item Run t-walk MCMC sampler for 10,000 iterations
                    \item Discard the first 10,000 iterations as burn-in
                    \item Apply thinning (retain every 5th draw) to yield 20,000 posterior samples
                \end{itemize}
            \STATE Compute posterior means for model parameters
        \ENDFOR
        \STATE Compute root mean squared error (RMSE) and posterior summary statistics (e.g., mean, SD) across the $R$ replications
    \ENDFOR
\ENDFOR

\RETURN Posterior means, standard deviations, and RMSEs for each model and sample size
\end{algorithmic}
\end{algorithm}

%%%%%%%%%%%%%%%%%%%%%%%%%%%%%%%%%%%%%%%%%%
%%%%%%%%%%%%%%%%%%%%%%%%%%%%%%%%%%%%%%%%%%

\section{Additional Simulation Results}\label{add:sim.results}
\subsection{Damped hazard model}

%%%%%%%%%%%%%%%%%%%%%%%%%%%%%%%%%%%%%%%%%%
\begin{table}[H]
\centering
\caption{Posterior results for the critically damped hazard model. 
True parameter values are shown in parentheses.}
\label{tab:cdamp_ps}
\begin{tabular}{cccccc}
\toprule
 &  & $\alpha\,(2)$ & $\gamma\,(0.2)$ & $h_0\,(0.1)$ & $v_0\,(0.3)$ \\
\midrule
$n=200$  & Mean  & 2.9992 & 0.5213 & 0.1679 & 0.2120 \\
         & RMSE  & 1.3782 & 0.4591 & 0.0820 & 0.1901 \\
\midrule
$n=500$  & Mean  & 2.5740 & 0.3945 & 0.1572 & 0.2391 \\
         & RMSE  & 1.0176 & 0.3341 & 0.0672 & 0.1532 \\
\midrule
$n=1000$ & Mean  & 2.5427 & 0.3817 & 0.1258 & 0.2625 \\
         & RMSE  & 0.9691 & 0.3151 & 0.0388 & 0.1186 \\
\midrule
$n=2000$ & Mean  & 2.4550 & 0.3380 & 0.1081 & 0.2940 \\
         & RMSE  & 0.7853 & 0.2500 & 0.0216 & 0.0881 \\
\midrule
$n=5000$ & Mean  & 1.9391 & 0.2017 & 0.1004 & 0.2803 \\
         & RMSE  & 0.4382 & 0.1110 & 0.0129 & 0.0635 \\
\bottomrule
\end{tabular}
\end{table}

%%%%%%%%%%%%%%%%%%%%%%%%%%%%%%%%%%%%%%%%%%

\begin{figure}[H]
    \centering
    \includegraphics[scale = 1]{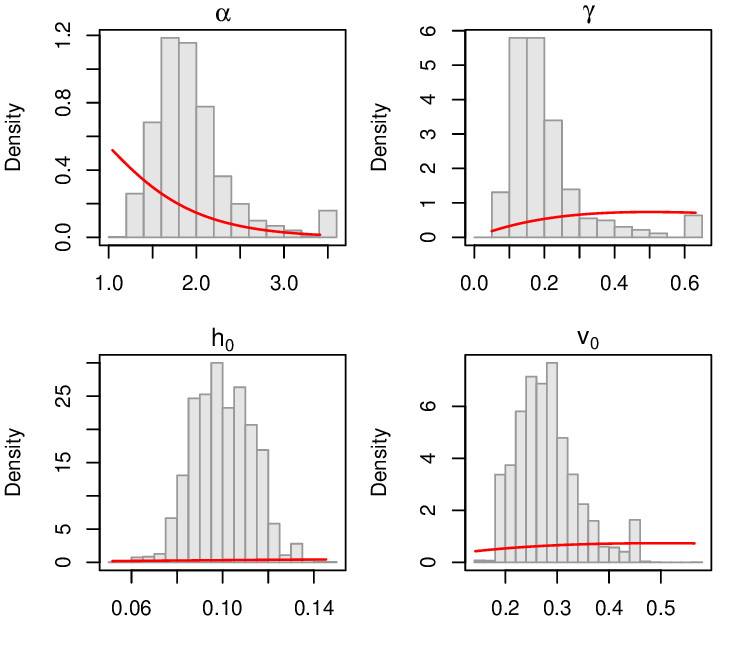}
    \caption{Posterior distributions (n=2000) for the critically damped hazard model parameters. The relevant part of the prior density is shown in red.}
    \label{fig:hist-c}
\end{figure}

%%%%%%%%%%%%%%%%%%%%%%%%%%%%%%%%%%%%%%%%%%

\begin{table}[H]
\centering
\caption{Posterior results for the overdamped hazard model. True parameter values are shown in
parentheses.}
\label{tab:odamp_ps}
\begin{tabular}{ccccccc}
\toprule
 &  & $\alpha\,(3)$ & $\beta\,(1)$ & $\gamma\,(0.2)$ & $h_0\,(0.1)$ & $v_0\,(0.3)$ \\
\midrule
$n=200$  & Mean  & 2.8125 & 1.1434 & 0.2326 & 0.0973 & 0.4313 \\
         & RMSE  & 0.8777 & 0.6084 & 0.1412 & 0.0519 & 0.2473 \\
\midrule
$n=500$  & Mean  & 2.7514 & 1.0504 & 0.2227 & 0.0842 & 0.2957 \\
         & RMSE  & 0.8793 & 0.6268 & 0.1408 & 0.0358 & 0.1350 \\
\midrule
$n=1000$ & Mean  & 2.9015 & 1.1208 & 0.2300 & 0.1024 & 0.2358 \\
         & RMSE  & 0.8047 & 0.6487 & 0.1330 & 0.0265 & 0.1349 \\
\midrule
$n=2000$ & Mean  & 2.8021 & 1.1172 & 0.2437 & 0.0985 & 0.2821 \\
         & RMSE  & 0.7568 & 0.6898 & 0.1572 & 0.0204 & 0.1029 \\
\midrule
$n=5000$ & Mean  & 2.4024 & 0.8495 & 0.1689 & 0.1034 & 0.2429 \\
         & RMSE  & 0.8751 & 0.5269 & 0.1115 & 0.0145 & 0.0987 \\
\bottomrule
\end{tabular}
\end{table}

%%%%%%%%%%%%%%%%%%%%%%%%%%%%%%%%%%%%%%%%%%

\begin{figure}[H]
    \centering
    \includegraphics[width=\linewidth]{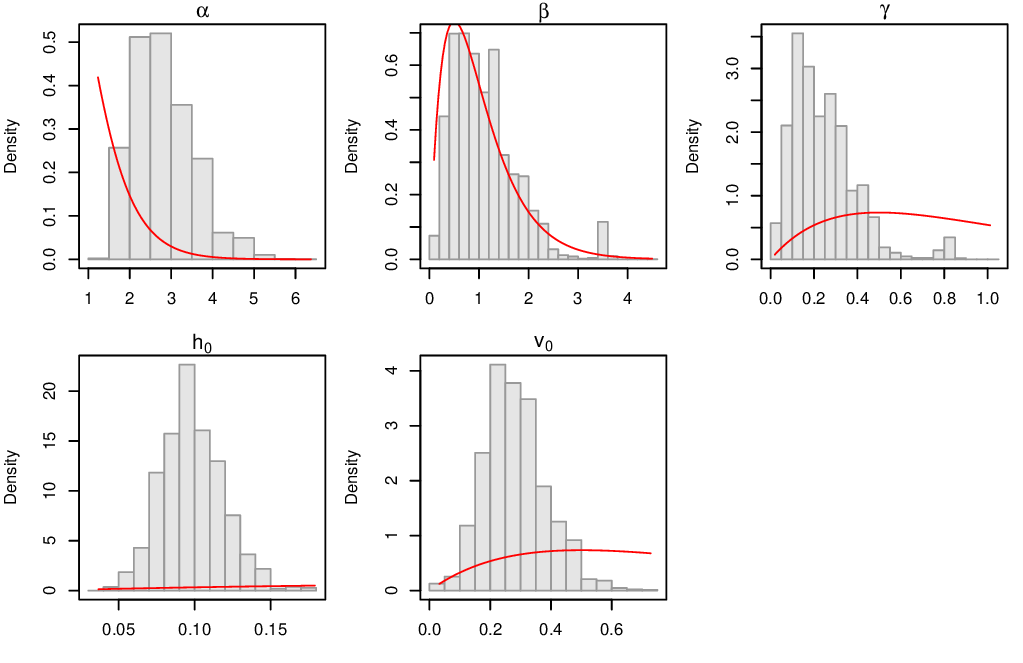}
    \caption{Posterior distributions (n=2000) for the overdamped hazard model parameters. The relevant part of the prior density is shown in red.}
    \label{fig:hist-o}
\end{figure}
%%%%%%%%%%%%%%%%%%%%%%%%%%%%%%%%%%%%%%%%%%

\subsection{Population dynamics–based hazard model}
\begin{table}[H]
\centering
\caption{Posterior results for the population dynamics-based hazard model. 
True parameter values are shown in parentheses.}
\label{tab:sinusoidal-ps}
\begin{tabular}{llccccc}
\toprule
 & & $r (0.8) $ & $\zeta (0.5)$ & $K (1)$ & $h_0 (0.1) $ & $v_0 (0.2)$ \\
\midrule
$n = 200$  & Mean  & 0.9774 & 0.9994 & 1.0240 & 0.0818 & 0.3728 \\
           & RMSE  & 0.7063 & 0.7218 & 0.4495 & 0.0452 & 0.2253 \\
\midrule
$n = 500$  & Mean  & 1.2490 & 0.7449 & 1.0382 & 0.0862 & 0.2479 \\
           & RMSE  & 0.9390 & 0.4499 & 0.4030 & 0.0315 & 0.1072 \\
\midrule
$n = 1000$ & Mean  & 1.2368 & 0.7126 & 0.9353 & 0.1046 & 0.1912 \\
           & RMSE  & 0.9311 & 0.4260 & 0.2457 & 0.0231 & 0.0642 \\
\midrule
$n = 2000$ & Mean  & 1.2947 & 0.7034 & 1.0430 & 0.0965 & 0.2201 \\
           & RMSE  & 0.9837 & 0.4189 & 0.1800 & 0.0175 & 0.0557 \\
\midrule
$n = 5000$ & Mean  & 1.0432 & 0.6564 & 0.9987 & 0.0950 & 0.2196 \\
           & RMSE  & 0.7000 & 0.3434 & 0.1185 & 0.0125 & 0.0386 \\
\bottomrule
\end{tabular}
\end{table}

%%%%%%%%%%%%%%%%%%%%%%%%%%%%%%%%%%%%%%%%%%
\begin{figure}[H]
    \centering
    \includegraphics[scale=1]{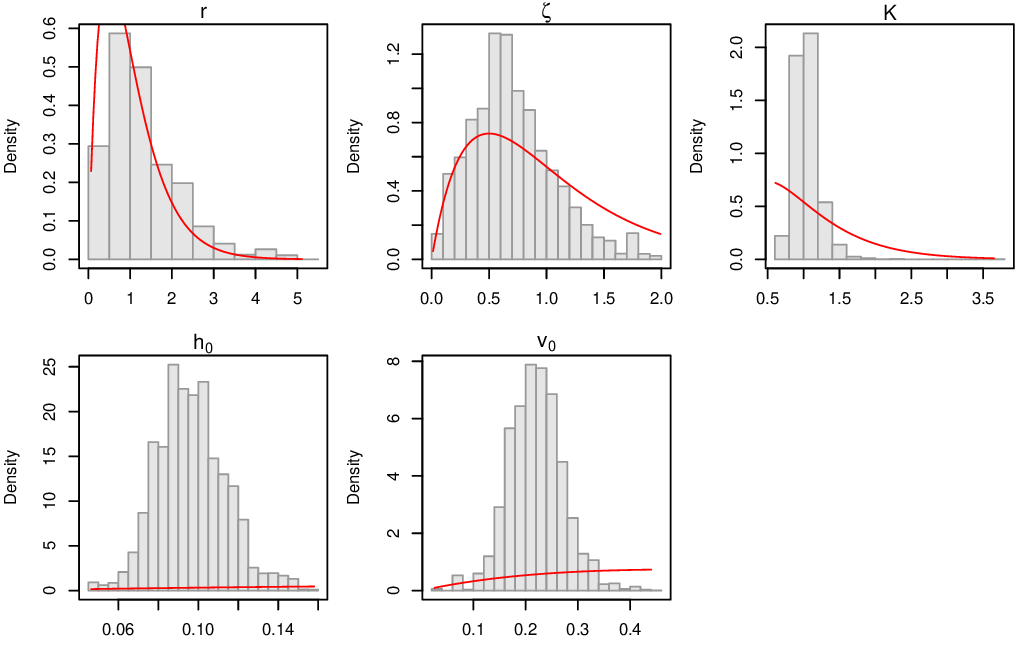}
    \caption{Posterior distributions (n=2000) for the population dynamics-based hazard model. The relevant part of the prior density is shown in red.}
    \label{fig:hist-expo0}
\end{figure}

\subsection{Sinusoidal hazard model}

%%%%%%%%%%%%%%%%%%%%%%%%%%%%%%%%%%%%%%%%%%

\begin{table}[H]
\centering
\caption{Posterior results for the sinusoidal hazard model. 
True parameter values are shown in parentheses.}
\label{tab:sinusoidal-ps}
\begin{tabular}{llcccc}
\toprule
 & & $\omega$ $(0.2\pi)$ & $c$ $(0.6)$ & $h_0$ $(0.1)$ & $v_0$ $(0.2)$ \\
\midrule
$n=200$  & Mean  & 0.7396 & 0.6326 & 0.1387 & 0.1641 \\
         & RMSE  & 0.1915 & 0.1444 & 0.0554 & 0.0608 \\
\midrule
$n=500$  & Mean  & 0.7568 & 0.5714 & 0.1028 & 0.1609 \\
         & RMSE  & 0.1565 & 0.0646 & 0.0230 & 0.0558 \\
\midrule
$n=1000$ & Mean  & 0.6928 & 0.6216 & 0.1048 & 0.1497 \\
         & RMSE  & 0.1114 & 0.0928 & 0.0184 & 0.0614 \\
\midrule
$n=2000$ & Mean  & 0.6723 & 0.6448 & 0.1056 & 0.1817 \\
         & RMSE  & 0.0741 & 0.0659 & 0.0140 & 0.0311 \\
\midrule
$n=5000$ & Mean  & 0.6432 & 0.6347 & 0.1035 & 0.1820 \\
         & RMSE  & 0.0396 & 0.0471 & 0.0092 & 0.0257 \\
\bottomrule
\end{tabular}
\end{table}

%%%%%%%%%%%%%%%%%%%%%%%%%%%%%%%%%%%%%%%%%%

\begin{figure}[H]
    \centering
    \includegraphics[scale=1]{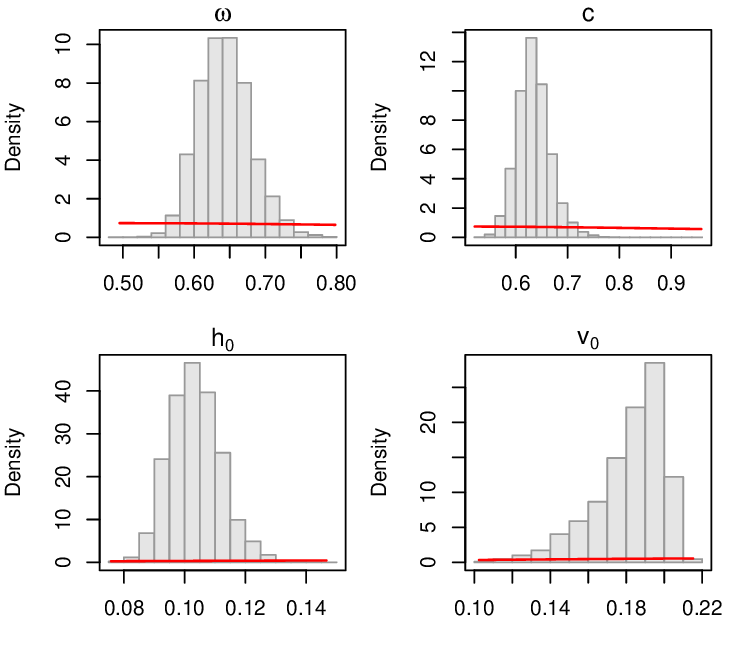}
    \caption{Posterior distributions (n=2000) for the sinusoidal hazard model parameters. The relevant part of the prior density is shown in red.}
    \label{fig:hist-sinusoidal}
\end{figure}
%%%%%%%%%%%%%%%%%%%%%%%%%%%%%%%%%%%%%%%%%%

\subsection{Exponential hazard model with interaction effects}

\begin{table}[H]
\centering
\label{tab:exp0-ps}
\caption{Posterior results for the exponential hazard model (when $\beta=0$). True parameter values are shown in parentheses.}
\begin{tabular}{llccc}
\toprule
 &  & $\alpha$ (0.1) & $h_0$ (0.4) & $v_0$ (0.1) \\
\midrule
$n = 200$  & Mean  & 0.1399 & 0.4415 & 0.0794 \\
           & RMSE  & 0.0693 & 0.0683 & 0.0387 \\
\midrule
$n = 500$  & Mean  & 0.1185 & 0.3873 & 0.0862 \\
           & RMSE  & 0.0399 & 0.0361 & 0.0268 \\
\midrule
$n = 1000$ & Mean  & 0.1186 & 0.3874 & 0.0941 \\
           & RMSE  & 0.0369 & 0.0283 & 0.0214 \\
\midrule
$n = 2000$ & Mean  & 0.1358 & 0.3952 & 0.1023 \\
           & RMSE  & 0.0456 & 0.0199 & 0.0200 \\
\midrule
$n = 5000$ & Mean  & 0.1056 & 0.3919 & 0.0980 \\
           & RMSE  & 0.0174 & 0.0148 & 0.0136 \\
\bottomrule
\end{tabular}
\end{table}

%%%%%%%%%%%%%%%%%%%%%%%%%%%%%%%%%%%%%%%%%%
\begin{figure}[H]
    \centering
    \includegraphics[scale=1]{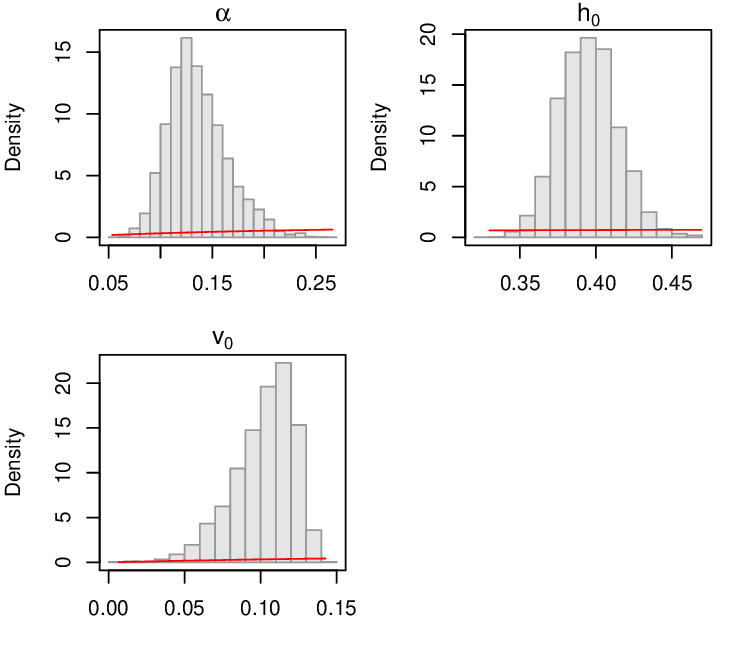}
    \caption{Posterior distributions (n=2000) for the exponential hazard model parameters (when $\beta=0$). The relevant part of the prior density is shown in red.}
    \label{fig:hist-expo0}
\end{figure}
%%%%%%%%%%%%%%%%%%%%%%%%%%%%%%%%%%%%%%%%%%

\end{document}